\documentclass[a4paper,onecolumn,accepted=2024-05-27]{quantumarticle}
\pdfoutput=1
\usepackage{amssymb,amsmath}

\usepackage{color}
\usepackage{comment}
\usepackage{mathtools}
\usepackage{graphicx}
\usepackage{physics}
\usepackage{xcolor}
\usepackage{microtype}
\usepackage[hyphens]{url}
\usepackage{authblk}
\usepackage{float}
\usepackage{enumitem}
\usepackage{amsthm}
\usepackage[backref=page,colorlinks,bookmarks=true]{hyperref}
\usepackage[capitalise,nameinlink]{cleveref}

\newtheorem{theorem}{Theorem}[section]
\newtheorem{lemma}[theorem]{Lemma}
\newtheorem{proposition}[theorem]{Proposition}
\newtheorem{definition}[theorem]{Definition}
\newtheorem{conjecture}{Conjecture}

\renewcommand{\backref}[1]{}

\renewcommand{\backrefalt}[4]{%
\ifcase #1 %
\or
[p.\ #2]%
\else
[pp.\ #2]%
\fi}

\Crefname{conjecture}{Conjecture}{Conjectures}

\DeclareMathOperator{\prob}{\mathrm{Pr}}
\DeclareMathOperator{\expect}{\mathrm{E}}
\DeclareMathOperator{\variance}{\mathrm{Var}}
\DeclareMathOperator{\negl}{\mathrm{negl}(\lambda)}
\DeclareMathOperator{\poly}{\mathrm{poly}(\lambda)}
\DeclarePairedDelimiter\ceil{\lceil}{\rceil}
\DeclarePairedDelimiter\floor{\lfloor}{\rfloor}
\newcommand*\diff{\mathop{}\!\mathrm{d}}

\def\Gen{\textnormal{Gen}}
\def\GenTrap{\textnormal{GenTrap}}
\def\Invert{\textnormal{Invert}}
\def\LWE{\textnormal{LWE}}
\def\MAJ{\textnormal{MAJ}}
\def\Encrypt{\textnormal{Encrypt}}
\def\Decrypt{\textnormal{Decrypt}}
\def\FirstResponse{\textnormal{FirstResponse}}
\def\SecondResponse{\textnormal{SecondResponse}}
\def\Guess{\textnormal{Guess}}

\begin{document}

\title{Lattice-Based Quantum Advantage from Rotated Measurements}

\author{Yusuf Alnawakhtha}
\affiliation{Joint Center for Quantum Information and Computer Science,
University of Maryland, College Park, MD 20742, USA}
\author{Atul Mantri}
\affiliation{Joint Center for Quantum Information and Computer Science,
University of Maryland, College Park, MD 20742, USA}
\affiliation{Department of Computer Science, Virginia Tech, Blacksburg, VA 24060, USA}
\author{Carl A.~Miller}
\affiliation{Joint Center for Quantum Information and Computer Science,
University of Maryland, College Park, MD 20742, USA}
\affiliation{Computer Security Division, National Institute of Standards and Technology, Gaithersburg, MD 20899, USA}
\author{Daochen Wang}
\affiliation{Joint Center for Quantum Information and Computer Science,
University of Maryland, College Park, MD 20742, USA}
\affiliation{Department of Computer Science, University of British Columbia, Vancouver, V6T 1Z4, Canada}

\maketitle

\begin{abstract}
Trapdoor claw-free functions (TCFs) are immensely valuable in cryptographic interactions between a classical client and a quantum server.  Typically, a protocol has the quantum server prepare a superposition of two bit strings of a claw and then measure it using Pauli-$X$ or $Z$ measurements.  In this paper, we demonstrate a new technique that uses the entire range of qubit measurements from the $XY$-plane.  We show the advantage of this approach in two applications. First, building on (Brakerski et al.~2018, Kalai et al.~2022), we show an optimized two-round proof of quantumness whose security can be expressed directly in terms of the hardness of the LWE (learning with errors) problem. Second, we construct a one-round protocol for blind remote preparation of an arbitrary state on the $XY$-plane up to a Pauli-$Z$ correction.
\end{abstract} 

\section{Introduction}

The field of quantum cryptography has its origins \cite{bennett1984,wiesner1983} in the idea that quantum states can be transmitted between two parties (e.g., through free space or through an optical fiber) to perform cryptographic tasks.
Properties of the transmitted states, including no-cloning and entanglement, are the basis for interactive protocols that enable a new and qualitatively different type of security.
However, a recent trend in the field has shown that quantum cryptography can be done even when quantum communication is not available. If one or more parties involved in a protocol possess a sufficiently powerful quantum computer, then certain cryptographic tasks can be performed --- while still taking advantage of uniquely quantum properties --- using strictly classical communication.  This approach relieves the users of the difficulties associated with reliable quantum communication and puts the focus instead on the problem of building a more powerful quantum computer, a goal that has seen tremendous investments during the past several years \cite{national2019quantum}.

At the center of this new type of quantum cryptography are cryptographic hardness assumptions.  Certain problems, such as factoring numbers, are believed to be difficult for classical computers but not for quantum computers.  Other problems, such as finding the shortest vector in a lattice, are believed to be hard for both types of computers.  These hardness assumptions are used to prove soundness claims for quantum interactive protocols.

Two of the seminal papers in quantum cryptography with classical communication \cite{mahadev2018classical,brakerski2018cryptographic} used \textit{trapdoor claw-free functions} \cite{goldwasser1984paradoxical} as the basis for their protocol designs, and created a model that has been followed by many other authors.  A trapdoor claw-free function (TCF), roughly speaking, is a family of triples $(f_0, f_1, t)$, where $f_0$ and $f_1$ are injective functions with the same domain and same range, and $t$ is a trapdoor that allows efficient inversion of either function.  To say that this family is \emph{claw-free} means that without the trapdoor $t$, it is believed to be hard for any (quantum or classical) adversary to find values $x_0$ and $x_1$ such that $f_0 ( x_0 ) = f_1 ( x_1)$.  

The TCF construction illustrates how a cryptographic hardness assumption that is made for both quantum and classical computers can nonetheless permit a quantum computer to show its unique capabilities.  A quantum computer can perform an efficient process that will output a random element $y$ in the range of $f_0, f_1$ together with a \emph{claw state} of the form
\begin{equation}
\label{eq:origpsi}
\ket{ \psi }  =  \frac{1}{\sqrt{2}} \left( \left| x_0 \right> \left| 0 \right> + \left| x_1 \right> \left| 1 \right> \right),
\end{equation}
where $f_0 ( x_0 ) = f_1 ( x_1 ) = y$ (see section~2 of \cite{mahadev2018classical}).  If this state is measured in the $Z$-basis, one obtains a pair $(x, c)$ such that $f_c ( x ) = y$.  Alternatively, assuming that $x_0$ and $x_1$ are expressed as bit strings of length $\ell$, and thus $\ket{\psi}$ is an $(\ell+1)$-qubit state, one can measure in the $X$-basis to obtain a bit string $d$ that must satisfy
\begin{equation}
\label{dcond}
d \cdot (x_0 \oplus x_1 || 1 ) = 0.
\end{equation}
(Here, $||$ denotes string concatenation.)  This equation is significant because we have used a quantum process to obtain information about both $x_0$ and $x_1$, even though we have assumed that it would be impossible for any efficient computer to recover $x_0$ and $x_1$ entirely.  
This fact is the basis for using TCFs to verify that a server that one is interacting with is able to perform quantum computation \cite{brakerski2018cryptographic,kahanamoku2022classically}. The same concept was also used in cryptographic constructions that delegate the preparation of a quantum state to a server without revealing its classical description \cite{cojocaru2019qfactory,gheorghiu2019computationally} and in other cryptographic protocols \cite{metger2021selftest,mahadev2018classical}.

The majority of papers utilizing TCFs in their cryptographic constructions have applied only Pauli measurements and classical operations to the state $\left| \psi \right>$.\footnote{Two exceptions are as~\cite{gheorghiu2019computationally} and \cite{cojocaru2021possibility}. In~\cite{gheorghiu2019computationally}, the server applies Fourier transforms to the quantum state $\left| \psi \right>$. In~\cite{cojocaru2021possibility}, the server applies measurements from a small set in the $XY$-plane and the protocol only provides security against honest-but-curious adversaries. See \cref{subsec:related} for a comparison.}
What would happen if we considered the full range of single-qubit measurements on the state $\ket{\psi}$?  We note that since single-qubit rotation gates are physically native in some platforms (for example, ion traps \cite{nielsenchuang2010book,monroe2016rotation,maslov2017rotation}), realizing a continuous single-qubit rotation is not much more difficult than realizing a single-qubit Clifford gate, and so this direction is a natural one to study. 

In this work, we use an infinite family of qubit measurements to prove new performance and security claims for quantum server protocols.  We discuss two applications: proofs of quantumness and blind remote state preparation.

\subsection{Our Contribution}
\label{subsec:contribution}

\paragraph*{Proof of Quantumness.}
With increasing efforts in recent years towards building quantum computers, the task of verifying the quantum behavior of such computers is particularly important.
Integer factorization, one of the oldest theoretical examples of a quantum advantage~\cite{shor1994algorithms}, is one possible approach to this kind of verification.
However, building quantum computers that are able to surpass classical computers in factoring is experimentally difficult and a far-off  task. Hence, it is desirable to find alternative \emph{proofs of quantumness}\footnote{By proof of quantumness, we mean a specific test, administered by a classical verifier, which an efficient quantum device can pass at a noticeably higher rate than any efficient classical device.} that are easier for a quantum computer to perform. 

The authors of \cite{brakerski2018cryptographic} did groundbreaking work in this direction by offering an interactive proof of quantumness  based on the hardness of LWE (learning with errors). Follow-up work \cite{brakerski2020simpler,kahanamoku2022classically,kalai2022quantum} used their technical ideas to optimize some aspects or provide trade-offs under different assumptions. In this work we provide a proof of quantumness that utilizes rotated measurements on claw states (\cref{eq:origpsi}) to achieve some new tradeoffs. The advantage achieved in our protocol by a quantum device is described in the following theorem.

\begin{theorem}[Informal]\label{thm:proof of quantumness}
Let $\lambda$ denote the security parameter. Suppose that $n, m, q, \sigma, \tau$ are functions of $\lambda$ that satisfy the constraints given in
\cref{fig:parameters} from \cref{sec:prelim}, and suppose that the $\LWE_{n,q,G ( \sigma, \tau )}$ problem is hard.  Then, there exists a two round interactive protocol between a verifier and a prover such that the following holds:
    \begin{itemize}
        \item For any efficient classical prover, the verifier accepts with probability at most $\frac{3}{4} + \negl$.
        \item For any quantum prover that follows the protocol honestly, the verifier accepts with probability at least $\cos^2\left(\frac{\pi}{8}\right) - \frac{5m\sigma^2}{q^2} - \frac{m\sigma}{2\tau}$.
    \end{itemize}
\end{theorem}

The protocol for this theorem is referred to as Protocol~$\mathbf{Q}$ and is given in \cref{fig:pfprot}.
Noting that $\cos^2 ( \pi / 8 ) > 0.85 > \frac{3}{4}$, we  deduce that as long as the error term $\frac{5m\sigma^2}{q^2} + \frac{m\sigma}{2\tau}$ vanishes as $\lambda \to \infty$, a constant gap is achieved between the best possible quantum winning probability and the best possible classical winning probability.  
In \cref{subsec:param}, we show that this vanishing condition can be achieved while taking the modulus $q$ to be only slightly asymptotically larger than $n^2 \sigma$ (where the parameter $n$ is the dimension of the LWE problem, and $\sigma$ is the noise parameter).
Our approach thus allows us to base the security of our protocol on the LWE problem for a broad range of parameters, including parameters that are comparable to those used in post-quantum lattice-based cryptography.
(For example, in the public-key cryptosystem in \cite{regev2009lwe}, the modulus is taken to be on the order of $n^2$.)
Previous works on interactive lattice-based proofs of quantumness have tended to use a modulus that is either very large or not explicitly given.

Our proof builds on recent work, and most directly builds on \cite{kalai2022quantum}. In comparison to the proof of quantumness described in \cite{kalai2022quantum}, our protocol involves preparing and measuring only one TCF claw state at each iteration, whereas the proof described in (\cite{kalai2022quantum}, Section~4.1) requires preparing and measuring three TCF claw states (while maintaining quantum memory throughout).  Additionally, whereas \cite{kalai2022quantum} requires the use of quantum homomorphic encryption schemes proved by other authors (and does not give explicit parameters) our proof is self-contained and directly relates the security of our protocol to the hardness of the underlying LWE problem.  At the same time, our approach inherits the following merits from \cite{kalai2022quantum}: our protocol involves only $2$ rounds of interaction, it requires only one qubit to be stored in memory in between rounds, and it does not require any cryptographic assumptions beyond LWE.

As far as we are aware, the combination of features in our work has not been achieved before (see \cref{subsec:related}), and our results thus bring the community closer to establishing the minimal requirements for a lattice-based proof of quantumness. As experimental progress continues~\cite{zhu2021interactive}, there is good reason to think that these proofs of quantumness may be realizable in the near future.

\paragraph*{Remote State Preparation.}
Remote state preparation (RSP) is a protocol where a computationally weak client delegates the preparation of a quantum state to a remote server. An RSP protocol is blind if the server does not learn a classical description of the state in the process of preparing it \cite{dunjko2012blind}. Recently, \cite{cojocaru2019qfactory} and \cite{gheorghiu2019computationally} introduced blind RSP with a completely classical client, based on the conjectured quantum hardness of the LWE problem. Blind RSP has become an essential subroutine to \emph{dequantize} the quantum channel in various quantum cryptographic applications including blind and verifiable quantum computations~\cite{broadbent2009universal,gheorghiu2019computationally,badertscher2020security}, quantum money~\cite{radian2022semi}, unclonable quantum encryption~\cite{gheorghiu2022quantum}, quantum copy protection~\cite{gheorghiu2022quantum}, and proofs of quantumness~\cite{morimae2022proofs}.

All previous RSP protocols prepare a single-qubit state $ \frac{1}{\sqrt{2}}(\ket{0} + e^{i \theta} \ket{1})$, where $\theta$ belongs to some fixed set $S \subseteq [0,2\pi)$. However, a common feature among all the schemes is that either the size of $S$ must be small, or the basis determined by $\theta$ is not fixed a priori.  Therefore, a natural question in this context is: \\

\emph{Can a completely classical client delegate the preparation of arbitrary single-qubit states to a quantum server while keeping the basis fixed?}  \\

\noindent Ideally, we would like to achieve this task in a single round of interaction. Note that the previous RSP protocols along with \emph{computing on encrypted data} protocols such as~\cite{broadbent2009universal,fisher2014quantum} can realize this task in two rounds of interaction. In this work, we provide a simple scheme for \emph{deterministic} blind RSP that achieves this task without incurring any additional cost compared to previous \emph{randomized} RSP schemes. Our protocol only requires one round of interaction to prepare any single qubit state in the $XY$-plane (modulo a phase flip). This is particularly helpful for applications which require a client to delegate an encrypted quantum input, as it gives the client more control over the state being prepared. The correctness and blindness of the protocol are summarized in the theorem below.

\begin{theorem}[Informal]\label{remote state preparation}
    Let $\lambda$ denote the security parameter. Suppose that $n, m, q, \sigma, \tau$ are functions of the security parameter $\lambda$ that satisfy the constraints given in
    \cref{fig:parameters} from \cref{sec:prelim} and $\tau \geq 2m\sigma$, and suppose that the $\LWE_{n,q,G ( \sigma, \tau )}$ problem is hard.  Then there exists a one-round client-server remote state preparation protocol such that for any $\alpha \in \mathbb{Z}_q$, the client can delegate the preparation of the state  $\ket{\alpha}=\frac{1}{\sqrt{2}}\left( \ket{0} + e^{i \, 2\pi \alpha / q} \ket{1} \right)$ with the following guarantees:
    \begin{itemize}
      \item If the server follows the protocol honestly, then they prepare a state $\ket{\beta}$ such that 
      $\bigl\|\ketbra{\alpha}- Z^b \ketbra{\beta}Z^b \bigr\|_1 \leq \frac{4\pi m \sigma}{q}$,
      where $\norm{\cdot}_1$ denotes the trace norm and $b\in \{0,1\}$ is a random bit that the client can compute after receiving the server's response.
      \item The server gains no knowledge of $\alpha$ from interacting with the client. 
    \end{itemize}
\end{theorem}

\paragraph*{Additional Applications and Future Directions.} The techniques in this paper invite application to other cryptographic tasks, including verifiable random number generation, semi-quantum money, and encrypted Hamiltonian simulation. In~\cref{sec:future directions}, we discuss possible future research directions for these three topics, and we offer some preliminary results. 

There is ample room for further optimization of our results on lattice-based proofs of quantumness.  One of the advantages of our proof of soundness for Protocol~$\mathbf{Q}$ (\cref{thm:pfsound}) is that it is based directly on the security of a simple lattice-based encryption scheme (\cref{fig:kem}). (This is in contrast to the proof of quantumness in \cite{brakerski2018cryptographic}, in which the soundness proof is based on the ``adaptive hardcore bit property'' and its relationship to the security of LWE encryption schemes is more indirect.)  A natural next step to improve our results would be to try to replace the encryption scheme in \cref{fig:kem} with a more efficient single-bit encryption scheme. In particular, a Ring-LWE or Module-LWE based encryption scheme could allow for substantially less quantum computation time for an honest prover in Protocol~$\mathbf{Q}$.

\subsection{Related Works}
\label{subsec:related}

\paragraph{Proof of quantumness.} The study of proofs of quantumness based on LWE was initiated by Brakerski et al. \cite{brakerski2018cryptographic} who proposed a four-message (two-round) interactive protocol between a classical verifier and a prover. Their protocol
also involves constructing only a single TCF claw state (like in our protocol), although it requires holding the entire claw state in memory between rounds, and it uses an exponentially large modulus.\footnote{One effect of using an exponentially large modulus is on hardness assumptions.  If we phrase our hardness asssumptions in terms of the shortest vector problem in a lattice, then \cite{brakerski2018cryptographic} assumes the hardness of sub-exponential approximation of the shortest vector, while in the current work we only assume that polynomial approximation is hard.  See \cref{subsec:param}.} Later, \cite{brakerski2020simpler} gave a two-message (one-round) proof of quantumness with a simpler and more general security proof, but at the cost of requiring the random oracle assumption.
More recently, \cite{kahanamoku2022classically} introduced a proof of quantummness with some of the same features as \cite{brakerski2020simpler}  without the random oracle assumption, but they require up to $6$ messages in their protocol ($3$ rounds of interaction). 
Both \cite{amos2020one} and \cite{yamakawa2022verifiable} present constructions of \textit{publicly verifiable} proofs of quantumness, albeit with different assumptions or models.
Further works have based proofs of quantumness on different assumptions \cite{morimae2022proofs}, optimized the depth required for implementing TCFs~\cite{liu2021depth,hirahara2021test}, and even achieved prototype experimental realization~\cite{zhu2021interactive}.

The paper \cite{kalai2022quantum} (which is the 
main predecessor to this work) presented a generic compiler that turns any non-local game into a proof of quantumness and gave an explicit scheme that only requires $4$ messages ($2$ rounds of interaction).  These results assume the existence of a quantum fully homomorphic encryption scheme that satisfies certain conditions (see Definition 2.3 in \cite{kalai2022quantum}), and the authors sketch proofs that two known schemes \cite{mahadev2018classical, brakerski2018quantum} 
satisfy such conditions. 

An alternative route to proving a result similar to  \cref{thm:proof of quantumness} would be to expand the aforementioned proof sketches (see the appendix of the full version of \cite{kalai2022quantum}) and use them to create a proof of quantumness with explicit parameters.  However, we expect this approach to be less optimal both in mathematical complexity and computational complexity.  Combining \cite{kalai2022quantum} with \cite{mahadev2018classical} would involve an exponentially sized modulus $q$, and would require preparing and measuring three claw-states during a single iteration of the central protocol.  Combining \cite{kalai2022quantum} with \cite{brakerski2018quantum} would also require preparing three states of size similar to claw states, and would involve performing a $q$-ary Fourier transform on each state.  In contrast, our approach involves only preparing a single-claw state and then performing single qubit operations on it.

\paragraph{Blind RSP.}  Remote state preparation over a classical channel comes in two security flavors --- blind RSP and blind-verifiable RSP (in this work, we give a protocol for the former). Such a primitive was first introduced in~\cite{cojocaru2021possibility} for honest-but-curious adversaries. This was later extended to fully malicious adversaries in \cite{cojocaru2019qfactory}, where the authors present two blind RSP protocols, one of which allows the client to delegate the preparation of a BB84 state  $\left(\{\ket{0}, \ket{1},\ket{+},\ket{-}\}\right)$ and the other allows the client to delegate the preparation of one of $8$ states in the $XY$-plane. The former protocol has the advantage of allowing the client to choose if the server is preparing a computational $\left( \{\ket{0}, \ket{1} \} \right)$ or a Hadamard $\left( \{\ket{+}, \ket{-} \} \right)$ basis state. However, it is not clear how to generalize the scheme to prepare quantum states from a large set while maintaining control over the choice of basis. Independently,~\cite{gheorghiu2019computationally} gives a blind-verifiable RSP scheme that generalizes~\cite{brakerski2018cryptographic}, where the blindness is based on the adaptive hardcore bit property. The protocol in~\cite{gheorghiu2019computationally} can prepare one of $10$ states: $8$ in the $XY$-plane and the two computational basis states. There is a natural way to generalize \cite{gheorghiu2019computationally} to prepare states of the form $\frac{1}{\sqrt{2}}(\ket{0}+\exp(i \, 2\pi x/q)\ket{1})$ with $x,q \in \mathbb{N}$. However, the naturally generalized protocol requires an honest quantum server to apply a Fourier transform over $\mathbb{Z}_q$ on the claw state,
whereas we only require the server to perform single-qubit gates on the claw state for any $q$. Moreover, the quantity $x$ in the prepared state  is randomly chosen in~\cite{gheorghiu2019computationally} whereas our protocol allows the client to choose $x$. More recently, ~\cite{morimae2022proofs} constructs an RSP protocol from different cryptographic assumptions (full domain trapdoor permutations); however, the blindness is only shown against classical adversaries.

Previous RSP protocols have proven to be immensely useful in several cryptographic applications, ranging from proofs of quantumness~\cite{morimae2022proofs} and verification of quantum computation~\cite{gheorghiu2019computationally,zhang2022classical} to computing on encrypted data~\cite{badertscher2020security,gheorghiu2022quantum}, secure two-party quantum computation~\cite{ciampi2020secure} and more. Finally, RSP protocols have been extended to self-testing protocols ~\cite{metger2021selftest,mizutani2022selftest,fwz2022selftest}. A self-testing protocol characterizes not only the state prepared but also the measurements made upon them.

\section{Technical Overview}
\label{sec:technical overview}

Our protocol involves performing rotated measurements on the claw states put forward by \cite{brakerski2018cryptographic} to steer the final qubit of the claw state into a specific form, while keeping this form secret by hiding it using LWE. Suppose that $n, m, q$ are positive integers.  The \textit{Learning With Errors} (LWE) hardness assumption implies that if a classical client chooses a uniformly random vector $s \in \mathbb{Z}_q^n$ and a uniformly random matrix $A \in \mathbb{Z}_q^{m \times n}$, and computes $v \coloneqq As + e$ where $e \in \mathbb{Z}_q^m$ is a small noise vector, then a quantum server  cannot recover $s$ from $(A, v)$.  Following \cite{brakerski2018cryptographic} the server can (for certain parameters) nonetheless approximately prepare a superposed \emph{claw state} of the form $\ket{\gamma} = \frac{1}{\sqrt{2}} \bigl( \ket{x_1}\ket{1} + \ket{x_0} \ket{0} \bigr),$
where $x_0 \in \mathbb{Z}_q^n$ and $x_1 = x_0 + s$, along with a vector $y \in \mathbb{Z}_q^m$ which is close to both $Ax_0$ and $Ax_1$.  We will assume that $x_0$ and $x_1$ are written out in base-$2$ using little-endian order.

At this point, rather than having the server measure $|\gamma\rangle$ in the $X$-basis, we can go in a different direction: suppose  that the client instructs the server to measure the $k^{\text{th}}$ qubit of $|\gamma\rangle$ in the basis $ (\cos \theta_k ) X + (\sin \theta_k ) Y$ where $\theta_k$ are real numbers for $k = 1, 2, \ldots, n \lceil \log q \rceil$, and report the result as a binary vector $u = (u_1, \ldots, u_{n \lceil \log q \rceil}).$
Once this is done, the state of the final remaining qubit will be $\frac{1}{\sqrt{2}} (\left| 0 \right> + e^{i \phi} \left| 1 \right>)$,
where
\begin{equation*}
\phi \coloneqq \langle \theta, [x_0]  \rangle - \langle \theta, [x_1] \rangle + \langle u, [x_0] \oplus [x_1] \rangle \cdot \pi.
\end{equation*}
Here $\langle  \cdot,\cdot \rangle$ denotes the dot product, and $[x_0]$ and $[x_1]$ denote the base-$2$ representations of $x_0$ and $x_1$.  
Since the quantum server cannot know both $x_0$ and $x_1$, they cannot compute $\phi$ from this formula.  However, if the client possesses a trapdoor to the original matrix $A$, then they can recover $x_0$ and $x_1$ from $y$ and compute $\phi$.

We can go further: if the client chooses a vector $t = (t_1, \ldots, t_n ) \in \mathbb{Z}_q^n$ and sets $\theta$ by the formula $\theta_{(i-1)n+j} \coloneqq 2^j t_i  \pi/q$ for $i \in \{ 1, 2, \ldots, n\}$ and $j \in \{ 1, 2, \ldots, \lceil \log q \rceil \}$, then
\begin{eqnarray*}
\phi &  = & \langle t, x_0 - x_1  \rangle \cdot (2 \pi / q ) +  \langle u, [x_0] \oplus [x_1]  \rangle \cdot \pi \\
& = & - \langle t, s \rangle \cdot (2 \pi / q ) +  \langle u, [x_0] \oplus [x_1]  \rangle \cdot \pi.
\end{eqnarray*}
The server thus computes a qubit that encodes a prescribed linear function $\langle t , s \rangle$ of $s$, modulo a possible phase flip that is known to the client -- see \cref{sec:rotated measurements}. 
At the same time, LWE-hardness guarantees that the vector $s$ remains unknown to the server.
(This can be seen an enhancement of the approach described in subsection~1.4 in \cite{cojocaru2021possibility}.  We have used a different measurement strategy in order to gain more control over the prepared qubit.)

We summarize our two applications of this idea, starting with blind remote state preparation (blind RSP).

\paragraph*{Blind RSP.}
We will denote the state that the client wants the server to prepare (modulo a phase flip) as $\ket{\alpha} \coloneqq \frac{1}{\sqrt{2}} \bigl(\ket{0} + e^{i \, 2\pi \alpha /q} \ket{1} \bigr)$ for some $\alpha \in \mathbb{Z}_q$ of the client's choice. The way we will ensure blindness with respect to $\alpha$ is by encrypting it using Regev's encryption scheme, described in \cref{subsec:kem}, and having the server use the ciphertext to prepare the state. To encrypt $\alpha$ as such the client requires, in addition to an LWE instance, a uniformly sampled random vector $f \leftarrow \{0,1\}^m$. The client then computes $f^\top (As+e) + \alpha$ and sends $(A, As+e)$ and $(a,w) \coloneqq \left(f^\top A, f^\top(As+e) + \alpha\right)$ to the server.

The server uses $(A, As+e)$ to create a claw state, which also yields the image $y$ of the claw, and then measures the claw state using the vector $a \in \mathbb{Z}_q^n$ as described above. Then the server rotates the final qubit around the $Z$-axis of the Bloch sphere by $2\pi w /q$. From the discussion above we see that the resulting state will be $\frac{1}{\sqrt{2}}\bigl(\ket{0} + e^{i\beta} \ket{1} \bigr)$ where
\begin{align*}
    \beta \coloneqq& - \langle  a , s  \rangle \cdot (2 \pi / q ) +  \langle  u, [x_0] \oplus [x_1] \rangle \cdot \pi + 2\pi w/q 
    \\
    =& - 2 \pi f^\top A s/ q  + 2\pi \bigl(f^\top (A s+e) + \alpha \bigr)/ q+  \langle  u, [x_0] \oplus [x_1] \rangle \cdot \pi
    \\
    \approx & \;2\pi \alpha / q+  b \cdot \pi,
\end{align*}
by denoting $b\coloneqq \langle  u, [x] \oplus [x_1] \rangle$ and letting $e/q$ be small (a more explicit calculation is provided in \cref{subsec:rsp completeness}). The final state held by the server is $Z^b \ket{\alpha}$, as desired. Finally, the server sends the measurements $y$ and $u$ to the client who can use them along with the trapdoor of $A$ to learn $b$. Note that while the client has control over $\alpha$, they do not have control over the bit $b$ as it is a function of the server's measurements.

The blindness property of the protocol is derived from the security of Regev's encryption scheme, which is based on the hardness of LWE. The information that the server receives during the protocol is $\left( A, As+e \right)$ and $\bigl(f^{\top}A, f^{\top}(As+e) + \alpha \bigr)$. The first pair is the public key in Regev's encryption scheme and the second pair is the ciphertext encrypting the message $\alpha$. Hence, if an adversary can guess $\alpha$ in our RSP protocol then they can break Regev's encryption scheme.

\paragraph*{Proof of Quantumness.} Our proof of quantumness is based on the  CHSH game~\cite{clauser1969proposed}, which is a game played with a referee and two players (Charlie and David) who cannot communicate with each other. The referee sends bits $b$ and $b'$ chosen uniformly at random to Charlie and David respectively. Charlie and David are required to return bits $d$ and $d'$ to the referee, who decides that they win if and only if $d\oplus d' = b\wedge b'$.

Recent work \cite{kalai2022quantum} has proposed a version of the nonlocal CHSH game that a single server (Bob) can play with a classical client  (Alice) to prove the quantumness of the server.  However, the protocol in \cite{kalai2022quantum} requires the server to evaluate a controlled-Hadamard gate under a homomorphic encryption scheme, which requires preparing and measuring three claw states while maintaining qubits in quantum memory. By combining the ideas in \cite{kalai2022quantum} with our RSP protocol, we obtain the following proof of quantumness protocol that only requires preparing and measuring one claw state.

Our central protocol (\cref{fig:pfprot} and \cref{fig:prover}) is roughly the following.
Alice begins by choosing bits $b$ and $b'$ uniformly at random.  Then, Alice uses our RSP protocol to delegate to Bob the preparation of a state of the form $(I + (-1)^d X)/2$ (if $b = 0$) or of the form $(I + (-1)^d Y )/2$ (if $b=1$).  From the RSP process, Bob obtains an encryption of the bit $d$ which he sends back to Alice for her to decrypt.  Alice then sends $b'$ to Bob, and he measures his state in the eigenbasis of $\frac{1}{\sqrt{2}} ( X + (-1)^{b'} Y )$ and returns the outcome bit $d'$ to Alice. Alice considers Bob to have won the game if and only if $d\oplus d' = b\wedge b'$.

It can be seen that the distribution over $(b,b',d,d')$ in the procedure described above is approximately the same as the distribution in the nonlocal CHSH game when the two players implement the optimal quantum strategy.  Therefore, a quantum Bob can win with probability approximately  $\cos^2(\pi/8)\approx 0.85$ (\cref{thm:comp}).
On the other hand, by adapting reasoning from~\cite{kalai2022quantum}, we show that an efficient classical Bob cannot do much better than the optimal classical winning probability for the CHSH game, which is $0.75$ (\cref{thm:pfsound}).  Therefore, with appropriate parameter choices (see \cref{subsec:param}), a constant gap is achieved between the best possible classical and quantum winning probabilities.

\section{Preliminaries}
\label{sec:prelim}

Let $\mathbb{C}, \mathbb{Z}, \mathbb{N}$ denote, respectively, the field of complex numbers, the ring of integers, and the set of nonnegative integers. 
For any $c \in \mathbb{N}$, let $\mathbb{Z}_c$ denote the set $\{ 0, 1, \ldots, c-1 \}$ with multiplication and addition defined modulo $c$.  If $x \in \mathbb{Z}_c$, then $
| x | \coloneqq \min \{ x, c - x \}.$
If $v = (v_1, \ldots, v_s)$ is a vector with entries from any of $\mathbb{C}, \mathbb{Z}$, or $\mathbb{Z}_c$, we write $\left\| v \right\|_\infty$ for the infinity norm of $v$, defined by $\norm{ v }_\infty \coloneqq \max_i | v_i |.$

We denote by $\{ 0, 1 \}^*$ the set of all finite-length bit strings.  For $s \in \{ 0, 1 \}^*$ and $k \in \mathbb{N}$, let $s^k \in \{ 0, 1 \}^*$ denote $s$ repeated $k$ times. The symbol $||$ denotes string concatenation. We write MAJ for the function $\MAJ \colon \{0,1\}^* \to \{0,1\}$ defined by $\MAJ(s) = 1$ if and only if the Hamming weight of $s$ is at least half its length. If $a,b$ are vectors of the same length, we may write either $\left< a , b \right>$ or $a \cdot b$ for the dot product of $a$ and $b$. 

For a finite set $S$, we write $s\leftarrow S$ to mean $s$ is sampled uniformly at random from $S$. If $\chi$ is a distribution on $S$, we write $s\leftarrow \chi$ to mean $s$ is sampled from $S$ according to $\chi$.  The expression $\chi^n$ denotes the distribution of $n$-length sequences of independent samples of $\chi$. For a set $A$, we write $D(A)$ for the set of probability distributions on $A$.

The expression $\log$ always denotes the logarithm in base $2$. If $k \in \mathbb{Z}_c$ (viewed as $\{0,1,\ldots, c-1\}$), then $[k]\in \{ 0, 1 \}^{\lceil \log c \rceil}$ denotes the binary representation of $k$ in little-endian order (i.e., with the least significant bits first).  If $x \in \mathbb{Z}_c^d$, then $[x]$ denotes the concatenation of $[x_1], [x_2], \ldots, [x_d]$.

For any finite set $S$, the expression $\mathbb{C}^S$ denotes the Hilbert space of functions from $S$ to $\mathbb{C}$.  Let $L ( S )$ denote the set of linear maps from $\mathbb{C}^S$ to itself. A \textit{quantum state} on $S$ is an element of $L ( S )$ that is trace-$1$ and positive semidefinite (i.e., a density operator on $\mathbb{C}^S$). A \textit{pure quantum state} on $S$ is a rank-$1$ quantum state. Any pure quantum state on $S$ can be written as $\rho \coloneqq \ketbra{\alpha}$, where $\ket{\alpha}$ is a unit vector in $\mathbb{C}^S$, and we may also refer to $\ket{\alpha}$ as a pure quantum state. The \textit{trace distance} between two quantum states $\rho$ and $\sigma$ on $S$ is defined by $\norm{\rho-\sigma}_1$, where $\norm{\cdot}_1$ denotes the trace norm. Note that $\norm{\ketbra{\alpha}-\ketbra{\beta}}_1 = 2 (1 - \abs{\braket{\alpha}{\beta}}^2)^{1/2}$. If $T$ is another finite set, then a \textit{quantum operation} from $S$ to $T$ is a completely positive trace-preserving map from $L(S)$ to $L(T)$.

Throughout this work, $\lambda \in \mathbb{N}$ denotes the security parameter. A function $f \colon \mathbb{N} \to \mathbb{R}_{\geq 0}$ of $\lambda$ is \textit{negligible} if it is $O ( \lambda^{-C})$ for all $C > 0$. 

\subsection{Model of Quantum Computation}

We define terms related to quantum algorithms. A standard way to define a quantum circuit is as a composition of gates drawn from some specified finite set of primitive gates. Since we are concerned here with protocols that involve general single-qubit rotations, we will use a larger set of primitive gates. Let $X$, $Y$, and $Z$ denote the Pauli operators
\begin{equation}
X = \begin{bmatrix}
0 & 1 \\ 1 & 0 \end{bmatrix}, \quad Y = \begin{bmatrix}
0 & -i \\ i & 0
\end{bmatrix},
\quad \text{and} \quad
Z = \begin{bmatrix}
1 & 0 \\ 0 & -1 \end{bmatrix}.
\end{equation}
A quantum circuit is then a composition of the following primitive operations:
\begin{enumerate}
\item Any single qubit gate of the form $e^{i \pi (p/q)R}$, where $p, q \in \mathbb{Z}$, $q \neq 0$, and $R \in \{ X, Y, Z \}$.

\item The Toffoli gate $T \colon \left( \mathbb{C}^2 \right)^{\otimes 3} \to 
\left( \mathbb{C}^2 \right)^{\otimes 3}$, given by $T \left| x \right> \left| y \right> \left| z \right> =
\left| x \right> \left| y \right> \left| z + xy \right>$.

\item The gate which creates a single qubit in state $\left| 0 \right>$.

\item The gate which measures a single qubit in the computational basis.
\end{enumerate}
If a quantum circuit is composed only of unitary gates (i.e., it does not involve creating or measuring qubits) then we refer to it as a \textit{unitary} quantum circuit.
If a quantum circuit $Q$ has $m$ input qubits, and $n$ output qubits, and $\ell$ intermediate primitive operations, then the size of $Q$ is $n + m + \ell$.  Such a circuit determines a function from
$\{ 0, 1 \}^m$ to the set of probability distributions on $\{ 0, 1 \}^n$.

\subsection{Learning With Errors}

For any real number $s > 0$, let $G( s )$ denote the discrete Gaussian probability distribution on $\mathbb{Z}$, defined as follows:
if $X$ is a random variable distributed according to $G ( s)$, and $x \in \mathbb{Z}$, then
\begin{equation}
P ( X = x )  =  \frac{ e^{- x^2/(2s^2)}}{\sum\limits_{y \in \mathbb{Z}} e^{- y^2/(2s^2) }}.
\end{equation}
If $t > 0$, let $G( s, t )$ denote the  distribution obtained from $G ( s )$
by conditioning on the event $\left| X \right| \leq t$. We will need the following lemma about $G(s)$.

\begin{lemma}[Corollary 9 of \cite{canonne2020gaussian}]\label{lem:gaussian_properties}
Let $X$ be distributed according to $G(s)$.   Then $\expect[X]=0$ and $\variance[X]\leq s^2$.
If $C\geq 0$, then $\prob [ X \geq C] \leq e^{-C^2/(2 s^2 )}$.
\end{lemma}

In the following, $\chi$ denotes a probability distribution on $\mathbb{Z}$ (which in this paper will always be a Gaussian or truncated Gaussian distribution).

\vskip0.2in

\textbf{The $\LWE_{c, d, \chi}$ Problem:} Let $\mathcal{D}_0$ denote the probability distribution of $(a, a \cdot s + e ) \in \mathbb{Z}_d^c \times \mathbb{Z}_d$, where $a \leftarrow \mathbb{Z}_d^c$, $s \leftarrow \mathbb{Z}_d^c$, and $e \leftarrow \chi$. Let $\mathcal{D}_1$ denote the probability distribution of $(a, v)$, where $a \leftarrow \mathbb{Z}_d^c$ and $v \leftarrow \mathbb{Z}_d$.  Given oracle access to $\mathcal{D}_b$, where $b \leftarrow \{ 0, 1 \}$, determine the value of $b$.

\vskip0.2in

\subsection{Parameters and assumptions}
\label{subsec:param}

Throughout this paper, we will assume $m = m(\lambda),n = n ( \lambda ),q = q ( \lambda),Q = Q ( \lambda),\sigma = \sigma ( \lambda), \tau = \tau ( \lambda )$ are real-valued functions of the security parameter $\lambda$ which satisfy all of the conditions in \cref{fig:parameters}.

\begin{figure}
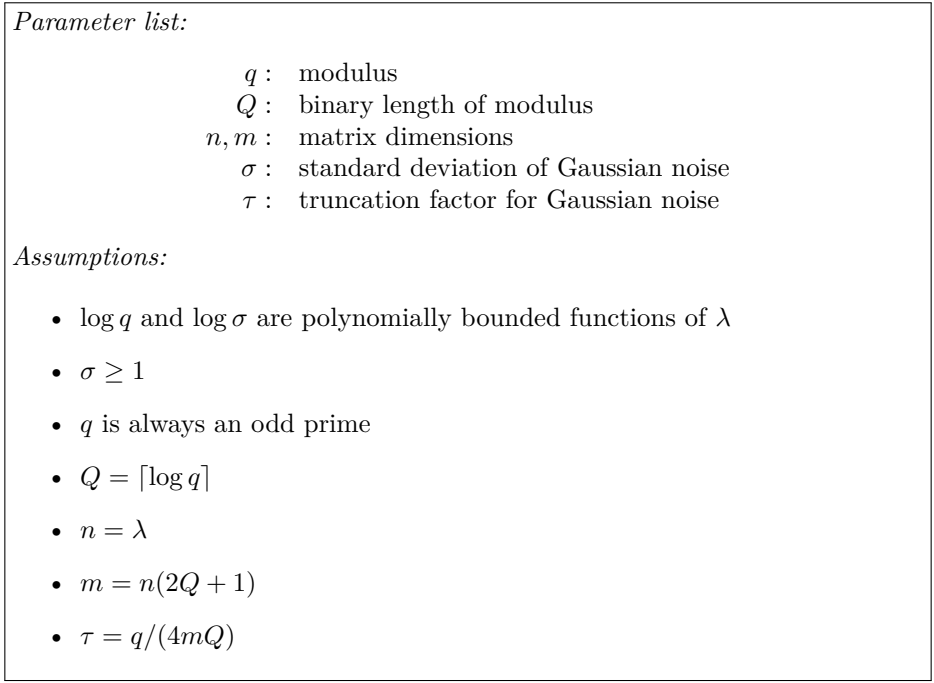

\centering
\fbox{\parbox{4.71in}{
\normalsize
\textit{Parameter list:}
\[
\begin{array}{rl}
q: & \textnormal{modulus} \\
Q: & \textnormal{binary length of modulus} \\
n,m: & \textnormal{matrix dimensions} \\
\sigma: & \textnormal{standard deviation of Gaussian noise} \\
\tau: & \textnormal{truncation factor for Gaussian noise} \\
\end{array}
\]
\smallskip
\textit{Assumptions:}
\begin{itemize}
\item $\log q$ and $\log \sigma$ are polynomially bounded
functions of $\lambda$

\item $\sigma \geq 1$

\item $q$ is always an odd prime

\item $Q = \lceil \log q \rceil$

\item $n = \lambda$

\item $m = n (2Q + 1 )$

\item $\tau = q/(4mQ)$
\end{itemize}
}}
\caption{ Parameters and assumptions.}
\label{fig:parameters}
\end{figure}

The rationale for the conditions in \cref{fig:parameters} is the following.
\begin{itemize}
    \item The quantity $m$, which specifies the number of rows in the LWE matrix $A$ that we will use, needs to be sufficiently large for the single-bit encryption algorithm in \cref{subsec:kem} to be secure and for $A$ to accommodate a trapdoor (\cref{subsec:trapdoor}).  The formula $m = n (2Q + 1)$ accomplishes both purposes.
    
    \item The truncation factor $\tau$ is chosen sufficiently small to allow LWE samples involving the matrix $A$ to be inverted, using a trapdoor, with probability $1$.
\end{itemize}
We discuss more specific parameter choices in \cref{subsec:parchoices}.

If we say that we assume that \textit{the $\LWE_{n, q, G ( \sigma )}$ problem is hard} for particular parameter functions $n = n ( \lambda ), q = q ( \lambda), \sigma = \sigma ( \lambda )$, we mean that we assume that any non-uniform quantum polynomial-time algorithm will solve the $\LWE_{n, q, G ( \sigma )}$ problem with probability at most $\frac{1}{2} + \negl$.  Note that if $\tau/\sigma = \omega ( \log \lambda )$, then the distribution of $G ( \sigma , \tau )$  is negligibly different from $G ( \sigma )$ by \cref{lem:gaussian_properties}, and 
so the hardness of $\LWE_{n, q, G ( \sigma, \tau )}$ is equivalent to the hardness of $\LWE_{n, q, G ( \sigma )}$.

As shown in \cite{regev2009lwe}\footnote{Strictly speaking, the next statement does not immediately follow from \cite{regev2009lwe} because the error distributions $\bar{\Psi}_\alpha$ and $\Psi_\alpha$ defined there do not exactly correspond to discrete Gaussians. However, it does follow after we first apply \cite[Theorem 1]{peikert2010reduction} to reduce $\LWE_{n,q,\Psi_{\sqrt{2\pi}\alpha}}$ to $\LWE_{n,q,G(\sigma)}$, where $\sigma = \alpha q$.}, we can assume that the $\LWE_{n, q, G (\sigma)}$ problem is hard for $n,q,\sigma = \alpha q$ with $\alpha \in (0,1)$ and  $\alpha q > \sqrt{2/\pi}\cdot \sqrt{n}$ if we assume that no non-uniform quantum polynomial-time algorithm can solve the Shortest Independent Vectors Problem (SIVP) in worst-case lattices of dimension $n$ to within an approximation factor of $\tilde{O}(n/\alpha)$.

\subsection{A Simple Encryption Algorithm}
\label{subsec:kem}

In \cref{sec:pf} we will make use of a single-bit encryption algorithm which is very similar to the original lattice-based encryption algorithm proposed by Regev in \cite{regev2009lwe}.   This algorithm is shown in \cref{fig:kem}, and consists of three algorithms: $\Gen$ (key generation), $\Encrypt$, and $\Decrypt$. Essentially, the public key is an LWE matrix $(A, v ) \coloneqq (A, As + e )$, where $s$ is a secret vector and $e$ is a Gaussian noise vector.  Given a random bit $b$, a ciphertext $ct$ is computed by summing up a random subset of the rows of $[A | v ]$ and then adding a quantity (dependent on $b$) to the last coordinate of the sum.  There is one important and unique feature of Protocol~$K$: rather than adding $b \lfloor q/2 \rfloor$ to the final coordinate of the ciphertext (which would optimize decoding) we add $b \lfloor q/4 \rfloor$ instead, which will aid us in \cref{sec:pf}.

\begin{figure}[ht]
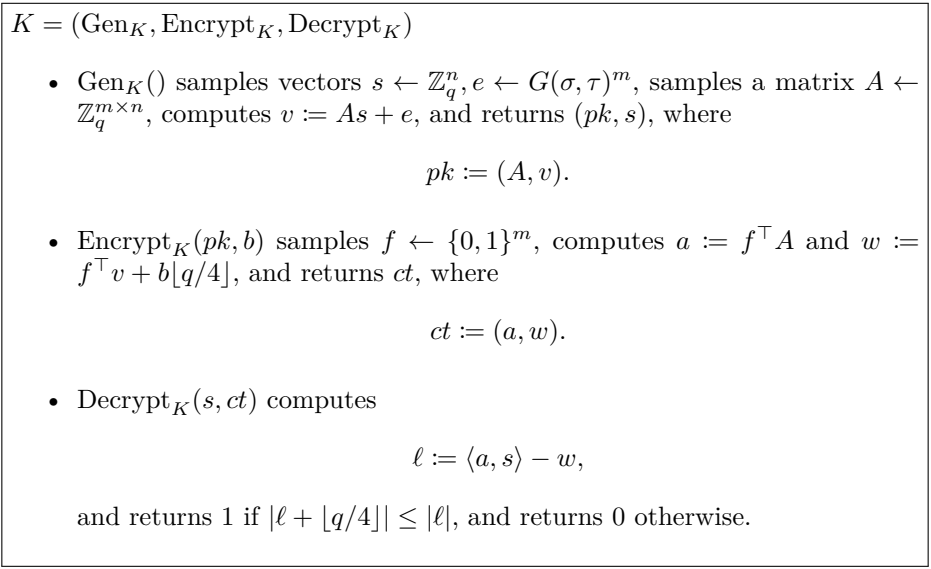

    \centering
    \fbox{\parbox{4.71in}{
    \normalsize
    $K = (\Gen_K, \Encrypt_K, \Decrypt_K)$

    \begin{itemize}
    \item $\Gen_K()$ samples vectors $s \leftarrow \mathbb{Z}_q^n, e \leftarrow G( \sigma, \tau)^m$, samples a matrix $A \leftarrow \mathbb{Z}_q^{m \times n}$, computes $v \coloneqq As + e$, and returns $(pk, s)$, where
    \begin{equation*}
    pk \coloneqq (A, v).  \\
    \end{equation*}

    \item $\Encrypt_K(pk, b)$ samples $f \leftarrow \{ 0, 1 \}^m$, computes
    $a \coloneqq f^\top A$ and $w \coloneqq f^\top v + b \lfloor q/4 \rfloor$, and returns $ct$, where
    \begin{equation*}
    ct \coloneqq (a, w ).
    \end{equation*}
    
    \item $\Decrypt_K(s, ct)$ computes
    \begin{equation*}
    \ell \coloneqq  \langle a , s \rangle - w,
    \end{equation*}
    and returns $1$ if $|\ell + \floor{q/4}| \leq |\ell|$, and returns $0$ otherwise.
    \end{itemize}
}}
    \caption{The single-bit public-key encryption algorithm $K$.  $pk$ is the public key, $s$ is the secret key, $b$ is the message, and $ct$ is the ciphertext.}
    \label{fig:kem}
\end{figure}

The following result asserts the IND-CPA security (that is, security against chosen-plaintext attacks) for Protocol $K$. The proof is standard and is given in \cref{app:regev}.

\begin{proposition}
\label{prop:regev}
If the $\LWE_{n, q, G ( \sigma , \tau )}$ problem is hard, then for any non-uniform quantum polynomial-time
algorithm $\mathcal{B}$, we have
$\prob [ b' = b ] \leq \frac{1}{2} + \negl$, where the probability is over $pk \leftarrow \Gen_K (),
b \leftarrow \{ 0, 1 \}, ct \leftarrow \Encrypt_K ( pk, b),$ and $b' \leftarrow \mathcal{B} ( pk, ct )$.
\end{proposition}

\subsection{Trapdoors for LWE matrices}
\label{subsec:trapdoor}

Both of the applications in this paper will rely on trapdoors for LWE samples. The following is a slightly modified version of
Theorem 2 from \cite{micciancio-peikert2012trapdoors}.  (The main difference is that we bound the noise vector using the infinity-norm rather than the Euclidean norm.)

\begin{proposition}
\label{prop:trap}
There is a probabilistic polynomial-time algorithm $\GenTrap()$ and a deterministic polynomial-time algorithm $\Invert(A, v, t)$ satisfying the following.
\begin{enumerate}
\item $\GenTrap()$ accepts no input and returns a pair $(A, t)$, where $A$ is an $m \times n$ matrix with entries in $\mathbb{Z}_q$.  The matrix $A$ is within statistical distance $nQ2^{-n/2}$ from a uniformly random matrix.

\item Given a pair $(A, t)$ obtained from $\GenTrap()$ and vectors $s \in \mathbb{Z}_q^n$ and $e \in \mathbb{Z}_q^m$ satisfying $\left\| e \right\|_\infty \leq 2 \tau$, the algorithm $\Invert ( A, As + e , t )$ returns the value $s$.
\end{enumerate}
\end{proposition}

\begin{proof}
See \cref{app:trap}.
\end{proof}

We make note of the following, which is an easy consequence of \cref{prop:trap}.

\begin{proposition}
\label{prop:veryinj}
If $(A, t)$ is a sample obtained from $\GenTrap()$, then for any nonzero vector $v \in \mathbb{Z}_q^n$, we must have $\left\| Av \right\|_\infty > 4 \tau$. 
\end{proposition}

\begin{proof}
Suppose that $v \neq 0$ were such that $\left\| Av \right\|_\infty \leq 4 \tau$.  Then, we can find vectors $e, e'$ of infinity norm less than or equal to $2 \tau$ such that $Av = e + e'$.  We have
\begin{equation}
0 = \Invert ( A, 0 + e , t ) 
= \Invert ( A, Av - e' , t )  = v,
\end{equation}
which is a contradiction.
\end{proof}

\begin{figure}
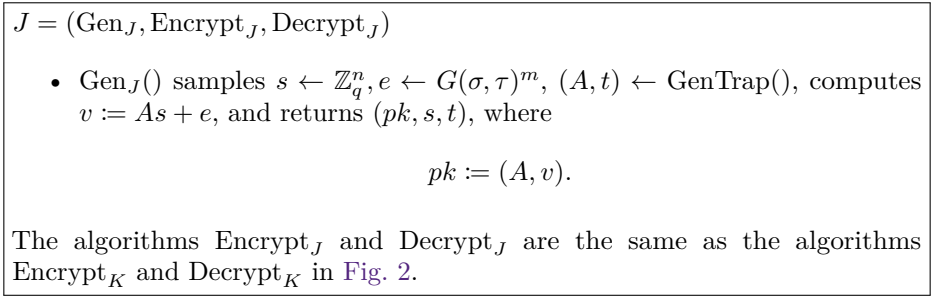

    \centering
    \fbox{\parbox{4.71in}{
    \normalsize
    $J = (\Gen_J, \Encrypt_J, \Decrypt_J)$

    \begin{itemize}
    \item $\Gen_J()$ samples $s \leftarrow \mathbb{Z}_q^n, e \leftarrow G( \sigma, \tau)^m$, $(A,t) \leftarrow \GenTrap()$, computes $v \coloneqq As + e$, and returns $(pk, s, t)$, where
    \begin{equation*}
    pk \coloneqq (A, v).  \\
    \end{equation*}
    \end{itemize}

The algorithms $\Encrypt_J$ and $\Decrypt_J$ are the same as the algorithms $\Encrypt_K$ and $\Decrypt_K$ in \cref{fig:kem}. 
}}
    \caption{The single-bit public-key encryption algorithm $J$.}
    \label{fig:trapkem}
\end{figure}

For the results in \cref{sec:pf}, it will be important to have a version of the encryption algorithm from \cref{fig:kem} that has a trapdoor for the encoding matrix $A$.  See \cref{fig:trapkem}. (Only the $\Gen()$ algorithm is different.  The trapdoor is not used for encryption or decryption.)
The following proposition follows directly from \cref{prop:regev,prop:trap}.

\begin{proposition}
\label{prop:regev2}
If the $\LWE_{n, q, G ( \sigma, \tau )}$ problem is hard, then 
for any non-uniform quantum polynomial-time
algorithm $\mathcal{B}$, we have $\prob [ b' = b ]
\leq  \frac{1}{2} + \negl$, where the probability is over $pk \leftarrow \Gen_J (), 
b \leftarrow \{ 0, 1 \}, ct \leftarrow \Encrypt_J ( pk, b ),$ and $ b' \leftarrow \mathcal{B}(pk,ct)$.
\end{proposition}

\section{Rotated Measurements on Generalized GHZ States}
\label{sec:rotated measurements}

The purpose of this section is to prove \cref{prop:last_qubit}, which shows how rotated measurements behave when applied to states that generalize GHZ states~\cite{ghz2007state}, defined below.
\begin{definition}[Generalized GHZ state]
Let $d$ be a positive integer. A \textit{generalized GHZ state} on $d+1$ qubits is a state of the form
\begin{equation}
\frac{1}{\sqrt{2}}(\ket{x}\ket{1} + \ket{y} \ket{0}),
\end{equation}
where $x,y\in \{0,1\}^d$.
\end{definition}

Given any sequence of real numbers $\theta_1, \ldots, \theta_n$, we make use of an associated sequence of real numbers $r_1, \ldots, r_{nQ}$ defined by 
\begin{equation}\label{eq:reqs}
\begin{array}{lllllllll}
r_1 = \theta_1 & &
r_2 = 2 \theta_1 & &
r_3 = 4 \theta_1 & &
\ldots &&  r_{Q}
= 2^{Q - 1} \theta_1 \\
r_{Q+1} = \theta_2 & &
r_{Q+2} = 2 \theta_2 & &
r_{Q+3} = 4 \theta_2 & &
\ldots &&  r_{2Q}
= 2^{Q - 1} \theta_2 \\
\vdots && \vdots && \vdots && \ddots && \vdots \\
r_{nQ-Q+1} = \theta_n & &
r_{nQ-Q+2} = 2 \theta_n & &
r_{nQ-Q+3} = 4 \theta_n & &
\ldots &&  r_{nQ}
= 2^{Q - 1} \theta_n.
\end{array}
\end{equation}

\begin{proposition}\label{prop:last_qubit}
Let $x,y \in \mathbb{Z}_q^n$, 
let $\ket{\psi}$ denote the $(nQ + 1 )$-qubit generalized GHZ state
\begin{equation}
\ket{ \psi } \coloneqq \frac{1}{\sqrt{2}} \bigl( \ket{ [x] } \otimes \ket{ 1 }
+ \ket{ [y] } \otimes \ket{ 0 } \bigr),
\end{equation}
and let $\theta_1, \ldots, \theta_n \in \mathbb{Z} \cdot (2\pi / q )$.  Let $r_1, \ldots, r_{nQ}$ be the sequence defined in \cref{eq:reqs}.  Suppose that for
each $i \in \{ 1, 2, \ldots, nQ \}$, the $i^\text{th}$ qubit of $\ket{\psi}$ is measured in the eigenbasis of
\begin{equation} \label{eq:measurement basis}
(\cos r_i ) X + (\sin r_i) Y 
\end{equation}
and that the outcome is $(-1)^{u_i} \in \{ -1, 1 \}$. 
Then the state of the remaining qubit is $\frac{1}{\sqrt{2}}\bigl(\ket{0} +  e^{i\theta} \ket{1}\bigr),
$ where
\begin{equation}\label{eq:phase of remaining qubit}
    \theta \coloneqq \sum_{i=1}^n (y_i - x_i)\theta_i + \pi \sum_{i=1}^n  \sum_{j=1}^Q ([y_i]_j - [x_i]_j) u_{(i-1)Q+j}.
\end{equation}
\end{proposition}

\begin{proof}
The eigenvectors of the matrix in \cref{eq:measurement basis} are given by
\begin{equation}
    \ket{+_{r_i}} \coloneqq \frac{1}{\sqrt{2}}(\ket{0} + e^{i r_i} \ket{1}) \quad \text{and} \quad 
    \ket{-_{r_i}} \coloneqq \frac{1}{\sqrt{2}}(\ket{0} - e^{i r_i} \ket{1}),
\end{equation}
where $\ket{+_{r_i}}$ is the $+1$ eigenvector and $\ket{-_{r_i}}$ is the $-1$ eigenvector. For a string $s\in \{0,1\}^{nQ}$ and $t\in \{1, 2, \ldots, nQ \}$, let $s_{\geq t}$ denote the suffix of $s$ starting at index $t$ (when $t> \abs{s}$,  $s_{\geq t}$ denotes the empty string).  Define $\phi_t$ by
\begin{equation}
    \phi_t \coloneqq
    \begin{cases}
         \sum_{j=1}^{t-1} ([y]_j-[x]_j)(r_j + \pi u_j) &\text{if $t >1$},
         \\
         0 &\text{if $t=1$}.
    \end{cases}
\end{equation}
We want to show that measuring the first qubit of the $nQ-t+2$ qubit state 
\begin{equation*}
\ket{\psi_t} \coloneqq \frac{1}{\sqrt{2}} \bigl( e^{i\phi_t}\ket{[x]_{\geq t}}\ket{1}+\ket{[y]_{\geq t}}\ket{0} \bigr)
\end{equation*}
in the eigenbasis of $\cos(\theta_t) X + \sin(\theta_t) Y$ yields the state $\ket{\psi_{t+1}}$. The description of the post-measurement state is given by
\begin{align}
    \nonumber&\frac{1}{\sqrt{2}}\bra{\pm_{r_t}}_1(e^{i\phi_{t}}\ket{[x]_{\geq t}}\ket{1}+\ket{[y]_{\geq t}}\ket{0}) \\
    =& \frac{1}{\sqrt{2}}(e^{i\phi_{t}}\braket{\pm_{r_t}}{[x]_t}\ket{[x]_{\geq{t+1}}}\ket{1}+\braket{\pm_{r_t}}{[y]_t}\ket{[y]_{\geq t+1}}\ket{0}) \\
    =& \frac{1}{\sqrt{2}} \left(  e^{-i [y]_t \left( r_t + \pi u_t \right)} \ket{[y]_{\geq t+1}}\ket{0} + e^{i \phi_t} e^{-i [x]_t \left( r_t + \pi u_t \right)} \ket{[x]_{\geq{t+1}}}\ket{1} \right)
    \\
    =& \frac{1}{\sqrt{2}} \left(  \ket{[y]_{\geq t+1}}\ket{0} + e^{i \left( \phi_t +  \left([y]_t - [x]_t \right) \left( r_t + \pi u_t \right)\right)} \ket{[x]_{\geq{t+1}}}\ket{1} \right)
    \\
    =& \frac{1}{\sqrt{2}}\left( \ket{[y]_{\geq t+1}}\ket{0} + e^{i \phi_{t+1}} \ket{[x]_{\geq{t+1}}}\ket{1} \right).
\end{align}

Since $\ket{\psi_1} = \ket{\psi}$, then after performing all of the $nQ$ measurements described in the proposition statement we are left with the state $\ket{\psi_{nQ+1}}$. It remains to show that $\phi_{nQ+1} = \theta$:
\begin{align}
    \phi_{nQ+1}
    =& \sum\limits_{j=1}^{nQ} ([y]_j-[x]_j)(r_j + \pi u_j)
    \\
    =& \sum\limits_{i=1}^{n}\sum\limits_{j=1}^{Q} ([y_i]_j-[x_i]_j)(r_{(i-1)Q+j} + \pi u_{(i-1)Q+j})
    \\
    =& \Bigl( \sum\limits_{i=1}^{n}\sum\limits_{j=1}^{Q} ([y_i]_j-[x_i]_j)2^{j-1}\theta_i \Bigr) + \Bigl(\pi \sum\limits_{i=1}^{n}\sum\limits_{j=1}^{Q}([y_i]_j-[x_i]_j) u_{(i-1)Q+j} \Bigr)
    \\
    =& \Bigl( \sum\limits_{i=1}^{n} (y_i-x_i)\theta_i \Bigr) +  \Bigl(\pi \sum\limits_{i=1}^{n}\sum\limits_{j=1}^{Q}([y_i]_j-[x_i]_j) u_{(i-1)Q+j} \Bigr) = \theta.
\end{align}
This completes the proof.
\end{proof}

\section{An Optimized Proof of Quantumness}
\label{sec:pf}

This section constructs a proof of quantumness based on the assumed hardness of the LWE problem. Our central protocol in this section, Protocol $\mathbf{Q}$ in \cref{fig:pfprot}, follows the form of the CHSH protocol from \cite{kalai2022quantum}, 
although \cite{kalai2022quantum} uses a quantum homomorphic encryption scheme and we instead use the encryption scheme $J$ from \cref{fig:trapkem}.
First the single prover is given an encrypted version of the first input bit $b$, and returns an encrypted version of the first output bit $d$ (steps 1--2).  Then, the prover is given the second input bit $b'$ as plaintext and returns $d'$ (steps 3--4).  Finally, the verifier decrypts $d$ (step 5) and scores the result (step 6).  

\begin{figure}
    \centering
    \fbox{\parbox{4.71in}{
    \normalsize
\textbf{Protocol~\textbf{Q}:}

\begin{enumerate}
    \item Alice samples $(pk = (A, v), s, t) \leftarrow \Gen_J()$, $b \leftarrow \{ 0, 1 \}$, and computes \[ ct = (a,w) \leftarrow \Encrypt_J( pk , b). \]  Alice broadcasts $(pk, ct)$ to Bob.
    
    \item \label{bobresponse} Bob returns a pair $(y, u)$ where $y \in \mathbb{Z}_q^m$ and $u \in \{ 0, 1 \}^{nQ}$.
    
    \item Alice samples $b' \leftarrow \{ 0, 1 \}$ and broadcasts $b'$ to Bob.
    
    \item Bob returns a bit $d'$.
    
    \item Alice computes $x_0 \coloneqq \Invert ( A, y, t)$
    and $x_1 \coloneqq \Invert ( A, y + v , t )$ and
    \begin{equation}
    \label{eq:hadamard string}
        d \coloneqq u \cdot ([x_0] \oplus [x_1]).
    \end{equation}

    \item If $d \oplus d' = b \wedge b'$, then Bob succeeds.  If $d \oplus d' \neq b \wedge b'$, then Bob fails.
\end{enumerate}
    }}
    \caption{The proof of quantumness protocol, including the behavior of the verifier (Alice).}
    \label{fig:pfprot}
\end{figure}

\subsection{Completeness}

The ideal behavior for the prover (Bob) is given in \cref{fig:prover}. The primary difference between Bob's strategy in \cref{fig:prover} in comparison to \cite{brakerski2018cryptographic,kahanamoku2022classically,kalai2022quantum} is the use of rotated measurements to compute $u$.

\begin{figure}[ht]
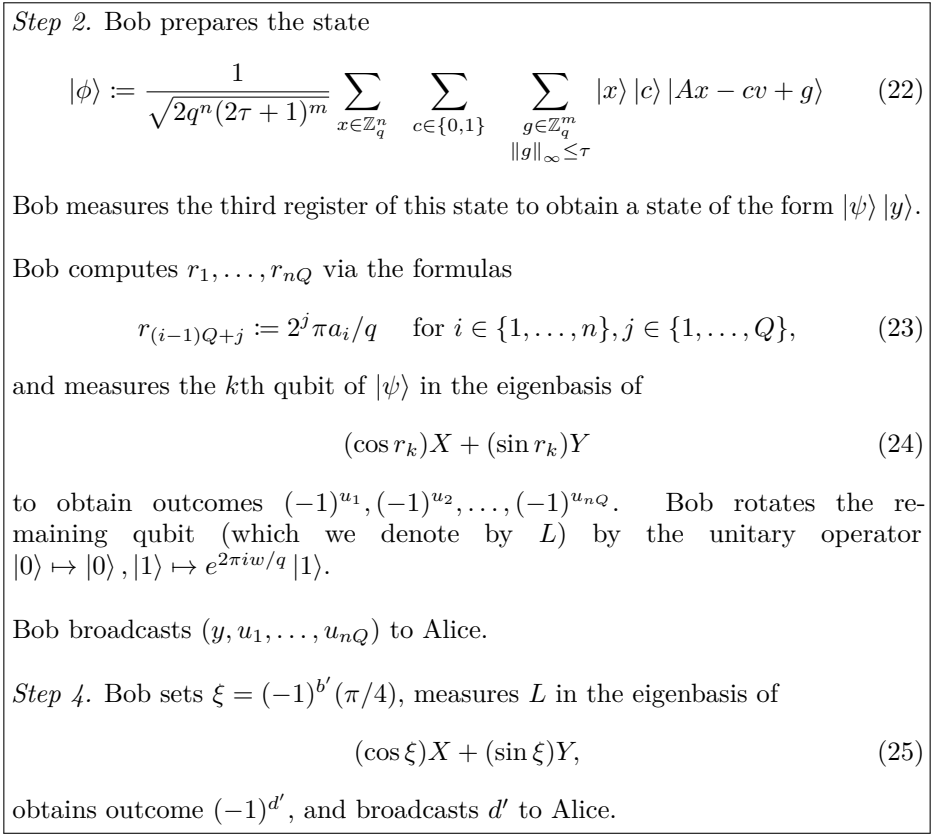

    \centering
    \fbox{\parbox{4.71in}{
    \normalsize
\textit{Step 2.}  Bob prepares the state
\begin{equation}
\ket{\phi} \coloneqq 
\frac{1}{\sqrt{2 q^n (2 \tau + 1 )^m}}
    \sum_{x \in \mathbb{Z}_q^n } \hskip0.1in \sum_{c \in \{ 0, 1 \}} \hskip0.1in
    \sum_{\substack{g \in \mathbb{Z}_q^m \\ \left\| g \right\|_\infty 
    \leq \tau }} \ket{x} \ket{c}
    \ket{Ax - cv + g}
\end{equation}
Bob measures the third register of this state to obtain a state of the form $\left| \psi \right> \left| y \right>$.  \\

Bob computes $r_1, \ldots, r_{nQ}$ via the formulas
\begin{equation}
r_{(i-1)Q + j} \coloneqq 2^j \pi a_i / q \quad \text{ for } 
i \in \{ 1, \ldots, n \}, j \in \{ 1, \ldots, Q \},
\end{equation}
and measures the $k$th qubit of $\left| \psi \right>$ in the eigenbasis of 
\begin{equation}
( \cos r_k ) X + ( \sin r_k ) Y
\end{equation}
to obtain
outcomes $(-1)^{u_1} , (-1)^{u_2}, \ldots, (-1)^{u_{nQ}}$.  Bob rotates the remaining qubit (which we denote by $L$) by the unitary operator $\ket{0} \mapsto \ket{0}, \ket{1} \mapsto e^{2 \pi i w / q} \ket{1}$.
 \\

Bob broadcasts $(y, u_1, \ldots, u_{nQ})$ to Alice. \\

\textit{Step 4.}  
Bob sets $\xi = (-1)^{b'}(\pi / 4 )$, measures $L$ in the eigenbasis of
\begin{equation}
(\cos \xi) X + (\sin \xi) Y,
\end{equation}
obtains outcome $(-1)^{d'}$, and broadcasts $d'$ to Alice.
}}
    \caption{The behavior of an honest quantum prover in Protocol~\textbf{Q} in \cref{fig:pfprot}.}
    \label{fig:prover}
\end{figure}

The optimal quantum score for the ordinary CHSH game is $\cos^2 (\pi/8) = (1/2 + \sqrt{2}/4)$.  The next theorem implies that an honest quantum prover will approach that score, provided that certain ratios between the parameters $q,m, \sigma$ and $\tau$ vanish as $\lambda$ tends to infinity.
\begin{theorem}
\label{thm:comp}
If Alice and Bob follow the process given in \cref{fig:pfprot,fig:prover}, 
then
\begin{equation}
\label{compexp}
\prob [ \textnormal{success} ] \geq \cos^2 \left( \frac{\pi}{8} \right) - \frac{5m \sigma^2}{q^2}  -
\frac{m \sigma}{ 2 \tau }.
\end{equation}
\end{theorem}

\begin{proof}
Let $T \coloneqq \{ e \in \mathbb{Z}_q \mid | e | \leq \tau \}$.  
Consider the map $S \colon \mathbb{Z}_q^n \times \{ 0, 1 \}
\times T^m \to \mathbb{Z}_q^m$ defined by $S ( x, c, g ) = Ax - cv + g$.

\cref{prop:veryinj} implies that any $y$ can have at most one pre-image in the set 
\begin{equation}
\label{eq:zerodomain}
\mathbb{Z}_q^n \times \{ 0 \}
\times T^m 
\end{equation}
and at most one pre-image in the set
\begin{equation}\label{eq:onedomain}
 \mathbb{Z}_q^n \times \{ 1 \}
\times T^m.
\end{equation}
Moreover, when $y$ has two pre-images $(x_0', 0, g_0)$ and $(x_1', 1, g_1)$ under $S$, we have
\begin{enumerate}
    \item $x_0' = x_0$  by \cref{prop:trap} since $\norm{Ax_0' - y}_\infty = \norm{g_0}_{\infty} \leq \tau$.
    \item $x_1' = x_1$ by \cref{prop:trap} since $\norm{Ax_1' - (y+v)}_\infty = \norm{g_1}_{\infty} \leq \tau$.
    \item  $x_1' = x_0' + s$ by \cref{prop:veryinj} since $\norm{A(x_1'- (x_0'+s))}_\infty \leq 2\tau + \norm{v-As}_\infty\leq 3\tau$.
\end{enumerate}

Let $U_0$ denote the set of all values of $y\in \mathbb{Z}_q^m$ that have a pre-image under $S$ in (\ref{eq:zerodomain}) and let $U_1$ denote the set of all values of $y\in \mathbb{Z}_q^m$ that have a pre-image under $S$ in set (\ref{eq:onedomain}).  A simple counting argument shows that, conditioned on the value of $e \coloneqq v - As \in \mathbb{Z}_q^m$, we have
\begin{equation}
\prob [ y \in U_0 \cap U_1 \mid e] = \frac{ \prod_{i=1}^m (2\tau +1 - |e_i|) }{ (2 \tau + 1 )^m }.
\end{equation}
Therefore we have the following, in
which we apply \cref{lem:abs}.
\begin{align}
\nonumber
\prob [ y \in  U_0 \cap U_1 ] & =  E \left[ \frac{ \prod_{i=1}^m (2\tau + 1 - | e_i |) }{ (2 \tau + 1 )^m } \Bigm| e \leftarrow G ( \sigma, \tau)^m \right] \\
\nonumber
& =  E \left[ \prod_{i=1}^m \left(  1 - \frac{   |e_i| }{ (2 \tau + 1 )^m } \right) \Bigm| e \leftarrow G ( \sigma, \tau)^m \right] \\
\nonumber
& \geq  E \left[ 1 - \sum_{i=1}^m \frac{   |e_i| }{ (2 \tau + 1 ) } \Bigm| e \leftarrow G ( \sigma, \tau)^m \right] \\
\nonumber
& =  1 - \frac{E \bigl[ \sum_{i=1}^m   |e_i|  \mid e \leftarrow G ( \sigma, \tau)^m \bigr] }{ (2 \tau + 1 ) } 
\\
&\geq  1 - \frac{m \sigma}{2 \tau + 1} \geq  1 - \frac{m \sigma}{2 \tau}.
\label{eq:validbound}
\end{align}

The following lemma upper bounds the probability of failure when $y\in U_0 \cap U_1$.
\begin{lemma}
\label{lem:bobprime}
If Alice and Bob follow the process in \cref{fig:pfprot,fig:prover}, then
\begin{equation}
\label{eq:noabort}
\prob [ \textnormal{failure}  \mid y\in U_0 \cap U_1]  \leq  
\sin^2 \left( \frac{\pi }{8} \right)
+ \frac{5 m \sigma^2}{q^2} + \frac{1}{q}.
\end{equation}
\end{lemma}

\begin{proof}
When $y\in U_0 \cap U_1$, by arguments made earlier in the proof, we have
\begin{equation}
    \ket{\psi} = \frac{1}{\sqrt{2}}(\ket{x_1}\ket{1} + \ket{x_0}\ket{0}),
\end{equation}
where $x_1 = x_0 + s$.

After Bob measures qubits $1, 2, \ldots, nQ$ of $\ket{\psi}$ to
obtain outcomes 
\begin{equation} 
(-1)^{u_1}, (-1)^{u_2}, \ldots, (-1)^{u_{nQ}}, 
\end{equation} 
the remaining qubit $L$ (by \cref{prop:last_qubit}) is in state $\frac{1}{\sqrt{2}} \bigl( \ket{0} + e^{i \theta} \ket{1} \bigr),$
where
\begin{equation}
\begin{aligned}
\label{eq:phase}
\theta & \coloneqq  \frac{2 \pi  a \cdot (x_0 - x_1) }{q}
+ \pi ([x_0] - [x_1]) \cdot (u_1, \ldots, u_{nQ})  \\
& =  - \frac{2 \pi (a\cdot s)}{q}
+ \pi ([x_0] \oplus [x_1]) \cdot (u_1, \ldots, u_{nQ}) \mod{2\pi}.
\end{aligned}
\end{equation}
Bob then rotates this state by the unitary operator $\ket{0} \mapsto \ket{0}, \ket{1}  \mapsto e^{i\, 2 \pi w / q} \ket{1}$, where $w = a\cdot s + f^\top e + b \floor{q/4}$, to obtain the state $\frac{1}{\sqrt{2}} \bigl( \ket{0} +
e^{i \beta} \ket{1} \bigr),$
where $\beta \coloneqq 
  \frac{2 \pi (f^\top e + b \lfloor q/4 \rfloor )}{q}
+ \pi ([x_0] \oplus [x_1]) \cdot (u_1, \ldots, u_{nQ})$.

Measuring this qubit with the observable $(\cos \gamma ) X + (\sin \gamma ) Y$, for any $\gamma \in \mathbb{R}$, yields outcome $+1$ with probability $\cos^2 ( (\gamma - \beta)/2)$ and outcome $-1$ with probability  $\sin^2 ( (\gamma - \beta)/2)$.
We therefore obtain the following formula for the failure probability:
\begin{equation}
\begin{aligned}
&\prob [ \textnormal{failure} \mid y\in U_0 \cap U_1 ] 
\\
&= \frac{1}{4} \expect \Biggl[ \sin^2 \left( \frac{ \pi f^\top e }{q} - \frac{\pi}{8} \right) +
\sin^2 \left( \frac{   \pi f^\top e }{q} + \frac{\pi}{8} \right) \\
&  \quad + \sin^2 \left( \frac{ \pi (f^\top e + \lfloor q/4 \rfloor ) }{q} - \frac{\pi}{8} \right) +
\cos^2 \left( \frac{ \pi  ( f^\top e +  \lfloor q/4 \rfloor) }{q} + \frac{\pi}{8} \right) \Biggr],
\end{aligned}
\end{equation}
where the expectation is over  $e \leftarrow G ( \sigma, \tau )^m$ and $f \leftarrow \{ 0, 1 \}^m$. For the rest of this proof, we consider all expectations to be over the conditions $e \leftarrow G ( \sigma, \tau )^m$ and $f \leftarrow \{ 0, 1 \}^m$.

We can obtain an upper bound on the above expression for $\prob [ \textnormal{failure} ]$ by replacing both instances of $\lfloor q / 4 \rfloor$ with $q/4$, and adding a term at the end of the expression to account for any increase that these replacements may cause.  Since the derivative of $\sin^2 ( \cdot )$ is always between $-1$ and $1$, inserting the term $(1/q)$ suffices. Therefore,
\begin{equation}
\begin{aligned}
\prob [ \textnormal{failure} \mid y\in U_0 \cap U_1]  & \leq 
\frac{1}{4} E \Biggl[ \sin^2 \left( \frac{ \pi f^\top e }{q} - \frac{\pi}{8} \right) +
\sin^2 \left( \frac{ \pi f^\top e }{q} + \frac{\pi}{8} \right) \\
&  + \sin^2 \left( \frac{ \pi f^\top e  }{q} + \frac{\pi}{8} \right) +
\cos^2 \left( \frac{ \pi  f^\top e  }{q} + \frac{3\pi}{8} \right) \Biggr]  + \frac{1}{q} \\
& = 4 \cdot \frac{1}{4} E \Biggl[ \sin^2 \left( \frac{ \pi f^\top e }{q} + \frac{\pi}{8} \right) \Biggr] + \frac{1}{q} . 
\end{aligned}
\end{equation}
Using \cref{lem:sin}, we get that
\begin{equation}
\prob [ \textnormal{failure} \mid y\in U_0 \cap U_1 ] \leq \sin^2 \left( \frac{\pi}{8} \right) + \sin \left( \frac{ \pi }{4 } \right) E \biggl[ \frac{ \pi f^\top e }{q} \biggr] +
 E \Biggl[ \left( \frac{  \pi f^\top e } {q}\right)^2  \Biggr] +  \frac{1}{q}.
\end{equation}
It is clear that  $\expect [f^\top e] = 0 $, while
\begin{equation}
    \expect[ (f^\top e)^2 ] = \sum_{i=1}^m \expect [ (f_i e_i)^2 ] =  \sum_{i=1}^m \frac{1}{2} \expect [ e_i^2 ] \leq m\sigma^2/2,
\end{equation}
where the first equality follows from \cref{lem:sum} and the last inequality follows from \cref{lem:var}. 

Therefore,
\begin{equation}
\prob [ \textnormal{failure} \mid y\in U_0 \cap U_1 ] \leq
\sin^2 \left( \frac{\pi}{8 } \right) +
\frac{\pi^2 \sigma^2 m}{2  q^2 } + \frac{1}{q} \leq
\sin^2 \left( \frac{\pi }{8} \right) +
\frac{5 \sigma^2 m}{  q^2 } + \frac{1}{q},
\end{equation}
as desired.
\end{proof}

Now we conclude the proof of \cref{thm:comp}. When the event $(y \in U_0 \cap U_1)$ does not occur, Bob merely measures a computational basis state on $nQ+1$ qubits using $nQ+1$ observables of the form $(\cos \gamma) X + (\sin \gamma) Y$. Therefore, the bits $u_1, \ldots, u_{nQ}$ and $d'$ that Bob returns to Alice are distributed uniformly at random, and Bob thus succeeds at Protocol $\mathbf{Q}$ with probability $1/2$.  Therefore, by \cref{lem:bobprime} and inequality (\ref{eq:validbound}), the overall probability that Bob fails in the protocol is upper bounded by
\begin{equation}
\sin^2 \left( \frac{\pi }{8} \right) + \frac{5m \sigma^2}{q^2} + \frac{1}{q} + \frac{1}{2} \cdot \frac{m \sigma}{2 \tau}.
\end{equation}
Since $\tau \leq q/4$ and $m \sigma \geq 1$, we can combine the last two summands
above to obtain
\begin{equation}
\prob[ \textnormal{failure} ] \leq
\sin^2 \left( \frac{\pi }{8} \right) + \frac{5m \sigma^2}{q^2} +  \frac{ m \sigma}{ 2 \tau},
\end{equation}
which implies the desired result.
\end{proof}

\subsection{Soundness}
\label{subsec:classical soundness}

The goal of this subsection is to put an upper bound on the probability that a classical prover could succeed at Protocol~\textbf{Q} in \cref{fig:pfprot}.  We model the behavior of a classical prover in \cref{fig:cheat}.  Bob responds to Alice's queries in two rounds using efficient classical algorithms while holding a private register $p$ in memory between the two rounds.

\begin{figure}[ht]
    \centering
    \fbox{\parbox{4.71in}{ 
    \normalsize
\textit{Step 2.}  Bob receives input $(pk, ct)$ and computes
\begin{equation*}
(y, u, p) \leftarrow \FirstResponse(pk, ct),
\end{equation*}
where FirstResponse is a non-uniform probabilistic polynomial-time algorithm.

\vskip0.2in

\textit{Step 4.} Bob receives $b'$ from Alice and computes
\begin{equation*}
d' \leftarrow \SecondResponse(b', p),
\end{equation*}
where SecondResponse is a non-uniform probabilistic polynomial-time algorithm.
}}
    \caption{A model of a classical adversary
    for Protocol~\textbf{Q} in \cref{fig:pfprot}. The register $p$ denotes internal memory held by the adversary between the rounds of communication.
}
    \label{fig:cheat}
\end{figure}

\begin{theorem}
\label{thm:pfsound}
Suppose that the $\LWE_{n,q,G ( \sigma, \tau )}$ problem is hard.  If Alice and Bob follow the process in \cref{fig:pfprot,fig:cheat}, then
\begin{equation}\label{soundexp}
\prob [ \textnormal{success}] \leq \frac{3}{4} + \negl.
\end{equation}
\end{theorem}

Our proof method comes from Section~3.1 of \cite{kalai2022quantum}.
We first make the following elementary observation.  Suppose that $T ( )$ is a non-uniform probabilistic polynomial-time algorithm that outputs a single bit, and that one wishes to optimally guess the output of $T ( )$.  Clearly, the highest probability with which this can be done is $\kappa \coloneqq \max \left\{ \xi, 1 - \xi \right\},$ where $\xi$ denotes the expected value of $T()$.
Consider the following procedure, which
uses the majority function (see \cref{sec:prelim}).
\begin{enumerate}
    \item Sample $z_1, z_2, \ldots, z_\lambda \leftarrow T ( )$.
    \item Output $\MAJ ( z_1, z_2, \ldots , z_\lambda )$.
\end{enumerate}
The Chernoff bound implies that if $z$ is obtained from this procedure and $z' \leftarrow T( )$ is a new sample, then $P ( z = z' )$ is  within $\exp ( - \Omega ( \lambda^{1/2} ))$ of the optimal guessing probability $\kappa$.

\begin{proof}[Proof of \cref{thm:pfsound}]
Suppose, for the sake of contradiction, that there exists a classical adversary 
\begin{equation*}
\mathcal{A} = (\FirstResponse,
\SecondResponse)
\end{equation*}
that achieves a success probability 
at Protocol~$\mathbf{Q}$ that is non-negligibly higher than $\frac{3}{4}$.  Consider
Experiment $\mathbf{C}$, shown in \cref{fig:chshmod}, in which two parties play a modified version of the CHSH game.
The Referee encrypts the first input bit $b$ using the scheme from Figure~\ref{fig:trapkem}, and transmits the resulting encryption to both Charlie and David while also giving Charlie the trapdoor $t$.  Then, Charlie and David play the CHSH game by doing a simulation of the behavior of the adversary $\mathcal{A}$.  The winning probability in Experiment~$\mathbf{C}$ is the same as the success probability for $\mathcal{A}$ in Protocol~$\mathbf{Q}$.

\begin{figure}[ht]
    \centering
    \fbox{\parbox{4.71in}{ \normalsize
\textbf{Experiment} $\mathbf{C}$: \\

Participants: Referee, Charlie, David

\begin{enumerate}
\item Referee chooses input bits $b, b' \leftarrow \{ 0, 1 \}$.  She computes $(pk, s, t) \leftarrow \Gen_J ()$ and $ct \leftarrow \Encrypt_J (pk, b)$ and sends $(pk, ct)$ to Charlie and David.  The referee also sends the trapdoor $t$ to Charlie.  

\item David computes $(y, u, p ) \leftarrow \FirstResponse( pk, ct )$ and shares $(y, u, p)$ with Charlie.

\item Referee transmits $b$ to Charlie and transmits $b'$ to David.

\item David computes $d' \leftarrow \SecondResponse( b', p )$ and transmits $d'$ back to the Referee.

\item \label{Charlie output} Charlie computes $x_0 \coloneqq \Invert ( A, y, t)$
    and $x_1 \coloneqq \Invert ( A, y-v , t )$ and
    \begin{equation}
    d = u \cdot ([x_0] \oplus [x_1]).
    \end{equation}
Charlie transmits $d$ to the Referee.

\item If $d \oplus d' = b \wedge b'$, then Charlie and David win; if not, they lose.
\end{enumerate}
}}
\caption{Two players (Charlie and David) play a modified version of the CHSH game using procedures FirstResponse and SecondResponse from \cref{fig:cheat}.}
\label{fig:chshmod}
\end{figure}

Note that if at step~\ref{Charlie output} Charlie was able to optimally guess the output of the probabilistic algorithm SecondResponse, then he could compute the output $d$ that maximizes their winning probability. Since this strategy would have the largest winning probability, we know that the winning probability of Experiment~$\mathbf{C}$ must be bounded above by it. Now consider Experiment~$\mathbf{C}'$,
shown in \cref{fig:chshmod2}, which has two changes from Experiment~$\mathbf{C}$.  First, Charlie is not given the trapdoor $t$ at step~\ref{refchoice}.  Also, at step~\ref{sampproc}, rather than attempt to compute the output bit $d$ directly, Charlie performs sampling to estimate the response $d$ that will maximize Charlie and David's winning probability. The winning probability in Experiment~$\mathbf{C}'$ can be at most negligibly lower (specifically, no more than $\exp ( - \Omega ( \lambda^{1/2} ))$ lower) than that where Charlie is able to optimally guess the output of SecondResponse, and thus it is at most negligibly lower than the winning probability in Experiment~$\mathbf{C}$.  Therefore, the winning probability in Experiment $\mathbf{C}'$ is also non-negligibly higher than $\frac{3}{4}$.

\begin{figure}[ht]
    \centering
    \fbox{\parbox{4.71in}{ 
    \normalsize
\textbf{Experiment} $\mathbf{C}'$: \\

Participants: Referee, Charlie, David

\begin{enumerate}
\item \label{refchoice} Referee chooses input bits $b, b' \leftarrow \{ 0, 1 \}$.  She computes $(pk, s, t) \leftarrow \Gen_J ()$ and $ct \leftarrow \Encrypt_J (pk, b)$ and sends $(pk, ct)$ to Charlie and David.

\item David computes $(y, u, p ) \leftarrow \FirstResponse( pk, ct )$ and shares $(y, u, p)$ with Charlie.

\item Referee transmits $b$ to Charlie and transmits $b'$ to David.

\item David computes $d' \leftarrow \SecondResponse( b', p )$ and transmits $d'$ back to the Referee.

\item \label{sampproc} Charlie samples $b'_1, \ldots, b'_\lambda \leftarrow \{ 0, 1 \}$, samples $d'_k \leftarrow \SecondResponse( b'_k, p )$ for each $k \in \{ 1, 2, \ldots, \lambda \}$, and computes
\begin{equation}
d \coloneqq \MAJ \left\{  d'_k \oplus (b \wedge b'_k) \mid k \in \{ 1, 2, \ldots, \lambda \} \right\}.
\end{equation}
Charlie transmits $d$ to the referee.

\item If $d \oplus d' = b \wedge b'$, then Charlie and David win; if not, they lose.
\end{enumerate}
}}
\caption{Experiment $\mathbf{C}'$ is the same as Experiment $\mathbf{C}$, except for steps \ref{refchoice} and \ref{sampproc}.}
\label{fig:chshmod2}
\end{figure}

Lastly, let Experiment $\mathbf{C}''$ denote a modified version of the Experiment $\mathbf{C}'$ in which, at step~\ref{refchoice}, the Referee generates
$ct$ via the procedure $ct \leftarrow \Encrypt_J ( pk , 0 )$ instead of $ct \leftarrow \Encrypt_J ( pk, b )$.  In Experiment $\mathbf{C}''$, Charlie and David are playing the original version of the CHSH game, in which both must compute their own outputs without any information about the other player's inputs.  In this case, we know that Charlie and David cannot win with probability more than $\frac{3}{4}$.  Therefore, the winning probabilities in Protocol $\mathbf{C}'$ and Protocol $\mathbf{C}''$ are non-negligibly different.  But this is a contradiction because it provides an efficient way to distinguish the probability distributions
\begin{equation}
[ (b, ct, pk) \mid 
(pk, s, t ) \leftarrow \Gen_J ( ), b \leftarrow \{ 0, 1 \}, ct \leftarrow \Encrypt_J (pk, b )] 
\end{equation}
and 
\begin{equation}
[ (b, ct, pk) \mid 
(pk, s, t ) \leftarrow \Gen_J ( ),  b \leftarrow \{ 0, 1 \}, ct \leftarrow \Encrypt_J (pk, 0 )] 
\end{equation}
which violates \cref{prop:regev2}.
\end{proof}

\subsection{Parameter Choices}
\label{subsec:parchoices}

Let $c, \epsilon$ be constant positive real numbers, and let the functions $q = q ( \lambda)$ and $\sigma = \sigma ( \lambda )$ be such that $\sigma$ is equal to $n^c$ and $q$ is an odd prime number between $n^{2+\epsilon} \sigma$ and $2 n^{2+\epsilon} \sigma$.
(See \cref{fig:parameters} for the definitions of the other parameters.)
Then,
\begin{equation*}
q = \Theta ( n^{2 + \epsilon + c } ), \quad Q  =  \Theta ( \log n ), \quad  m  =  \Theta ( n \log n ), \quad \tau  =  \Theta ( n^{1 
+ c + \epsilon} / (\log n )^2 ).
\end{equation*}
The lower bound (\ref{compexp}) from the previous completeness theorem for Protocol~$\mathbf{Q}$ then satisfies
\begin{align*}
&\cos^2 \left( \frac{\pi}{8} \right) - \frac{5m \sigma^2}{q^2}  -
\frac{m \sigma}{ 2 \tau } \\
\geq & \cos^2 \left( \frac{\pi}{8} \right) - O \left( \frac{ (n \log n )(n^{2 c})}{n^{4+2c + 2\epsilon}} \right)
- O \left( \frac{ (n \log n )(n^c)}{n^{1+ c + \epsilon} (\log n)^{-2}} \right) \\
= & \cos^2 \left( \frac{\pi}{8} \right) - O \left( \frac{\log n}{n^{3+2 \epsilon }} \right)
- O \left( \frac{ (\log n )^3}{n^\epsilon} \right),
\end{align*}
which tends to $\cos^2 ( \pi / 8)$ as  $n$ tends to infinity.  Meanwhile, assuming that $\LWE_{n, q , G ( \sigma )}$ is hard (noting that $G( \sigma, \tau)$ is negligibly different from $G ( \sigma )$ with these parameters), the upper bound in inequality (\ref{soundexp}) applies and tends to $\frac{3}{4}$ as $n$ tends to infinity.  Therefore, a constant gap is achieved between the best quantum success probability and our upper bound on the classical success probability.

\section{Blind Single-Qubit State Preparation}
\label{sec:blindrsp}

This section constructs Protocol~\textbf{P} in \cref{fig:rsp protocol} which allows a classical client to instruct a quantum server to prepare a single-qubit state $\frac{1}{\sqrt{2}}Z^b(\ket{0} + e^{i\, 2\pi \alpha/q}\ket{1})$, where $\alpha \in \mathbb{Z}_q$ and $b\in \{0,1\}$. Here, $\alpha$ is chosen by the client but kept hidden from the server while $b$ is a random bit that depends on the server's measurement outcomes. The client can compute $b$ after receiving the server's response.

To prepare the desired state in a secure manner, a classical client, Alice, interacts with the quantum server, Bob, in the following way. First, Alice sends an encoding of the (partial) classical description of the quantum state to Bob. As in \cref{sec:pf}, we use ideas from the LWE-based public key encryption scheme introduced by Regev~\cite{regev2009lwe} to design an encoding. Upon receiving the public key and the encoded message, Bob performs a quantum circuit and a series of measurements on the resulting state to steer the final qubit as described in \cref{sec:rotated measurements}. Bob rotates the last qubit, according to the information sent by Alice, to prepare the desired state up to a Pauli-$Z$ pad. The Pauli-$Z$ padding is partially determined by the measurement outcomes on Bob's end. Finally, Bob sends the measurement outcomes to Alice and she computes the full classical description of the prepared state. Namely, Bob's message allows Alice to learn $b$.

Throughout this section, we assume $\tau \geq 2m\sigma$, which is required to bound  the completeness error of Protocol~\textbf{P}. In terms of the parameters $n,q,\sigma$, this means we require 
\begin{eqnarray}
q/\sigma & \geq & 8n^2(2 \ceil{\log q}+1)^2 \ceil{\log q}
\end{eqnarray}
(see \cref{fig:parameters}), which can be satisfied for, e.g., $q = \Omega(n^3)$ and $\sigma = O(1)$. 

\begin{figure}
\centering
  \fbox{\parbox{5.5in}{
\begin{flushleft}
\textbf{\large Protocol \textbf{P}:}
\\
\smallskip
\textbf{Input:} \hspace{1.65mm} Alice: $\alpha \in \mathbb{Z}_q$.
\\
\textbf{Output:} Alice: classical description of $\ket{\alpha,b}$ (\cref{eq:alpha_b}). Bob: $\ket{\beta}$ (\cref{eq:beta}).
\end{flushleft}

\begin{enumerate}
    \item Alice samples $(A, t) \leftarrow \textnormal{GenTrap}()$, $s \leftarrow \mathbb{Z}_q^n, e \leftarrow 
    G ( \sigma, \tau)^m$, and $f \leftarrow \{ 0, 1 \}^m$. Then Alice broadcasts $(A,v) \coloneqq (A, As + e)$ and $(a, w) \coloneqq (f^\top A, f^\top v - \alpha)$ to Bob.
\item  Bob prepares the state
\begin{eqnarray}
\ket{\phi} & = &    
\frac{1}{\sqrt{2}\sqrt{q^n(2\tau+1)^{m}}}
    \sum_{\substack{x \in \mathbb{Z}_q^n}} \hskip0.1in \sum_{c \in \{ 0, 1 \}} \hskip0.1in
    \sum_{\substack{g \in \mathbb{Z}_q^m \\ \left\| g \right\|_\infty 
    \leq \tau }} \ket{x}\ket{c}
    \left| Ax - cv + g \right>,
\label{rspprepstate}
\end{eqnarray}
Bob measures the third register of $\ket{\phi}$ to obtain a state of the form $\ket{\psi}\ket{y}$. Note that $\ket{\psi}$ is an $(nQ+1)$-qubit state.

\vskip0.1in

Bob computes $r_{(i-1)Q+j} \coloneqq 2^j\pi a_i/q$ for all $i\in [n]$, $j\in [Q]$, and $r \coloneqq 2\pi w/q$.

\vskip0.1in

For each $i \in [nQ]$, Bob measures the $i$th qubit of the state $\ket{\psi}$ in the eigenbasis of $(\cos r_i ) X + (\sin r_i) Y$ and obtains outcome $(-1)^{u_i}$. 
\vskip0.2in

Let $\ket{\psi'}$ denote the state of the last  (unmeasured) qubit of $\ket{\psi}$. Bob prepares
\begin{equation} \label{eq:beta}
\ket{\beta} \coloneqq
\begin{bmatrix}
1 & 0 \\ 0 & e^{-ir} 
\end{bmatrix} \ket{\psi'}. 
\end{equation}
Bob broadcasts the vector $y\in \mathbb{Z}_q^m$ and the bit string $u\in \{0,1\}^{nQ}$ to Alice.  
\item  Alice computes whether $y$ belongs to the set
    \begin{equation}
    \{Ax + g \mid x\in \mathbb{Z}_q^n, g \in \mathbb{Z}_q^m, \norm{g}_\infty \leq \tau\} \cap \{Ax - v + g \mid x\in \mathbb{Z}_q^n, g \in \mathbb{Z}_q^m, \norm{g}_\infty \leq \tau\}.
    \nonumber
    \end{equation}
    If not, Alice aborts. Otherwise, Alice computes $x_0 \coloneqq \textnormal{Invert} ( A, y, t ) \in \mathbb{Z}_q^n$,  $x_1 \coloneqq x_0 + s \in \mathbb{Z}_q^n$,  $z \coloneqq [x_0]\oplus [x_1] \in \{0,1\}^{nQ}$, and $b \coloneqq z\cdot u \in \{0,1\}$. Finally, Alice computes the classical description of the single-qubit state
    \begin{equation}
    \label{eq:alpha_b}
        \ket{\alpha, b} \coloneqq Z^b \,  \frac{1}{\sqrt{2}}( \ket{0} + e^{i \, 2\pi \alpha/q } \ket{1}).
    \end{equation}
\end{enumerate}}}
\caption{The $2$-message blind remote state preparation protocol, including the behavior of the client (Alice) and an honest server (Bob).}
    \label{fig:rsp protocol}
\end{figure}

\subsection{Completeness}
\label{subsec:rsp completeness}

The purpose of this subsection is to prove \cref{thm:rsp completeness} which shows that the state $\ket{\beta}$ produced by Protocol~\textbf{P} in \cref{fig:rsp protocol}  is close to the state $\ket{\alpha,b}$, whose classical description is held by Alice. We first prove a technical lemma.

\begin{lemma}\label{lem:no_abort_rsp}
If Alice and Bob follow the process in \cref{fig:rsp protocol}, then
\begin{equation}
    \expect[\norm{e}_1 \mid \textup{no abort}] \leq 2m\sigma,
\end{equation}
where the expectation is over the distribution on $(A,t,s,e,f,y)$ defined by \cref{fig:rsp protocol}.
\end{lemma}
\begin{proof}
Since $e\leftarrow G(\sigma,\tau)^m$, by \cref{lem:abs}, we have that $\expect[\norm{e}_1] = \expect[\sum_{i=1}^m \abs{e_i}] \leq m \sigma$. But,
\begin{equation}
\begin{aligned}
    \expect[\norm{e}_1] &= \expect[\norm{e}_1 \mid \text{abort}]\prob[\textup{abort}] + \expect[\norm{e}_1 \mid \text{no abort}]\prob[\textup{no abort}]
    \\
    &\geq \expect[\norm{e}_1 \mid \text{no abort}]\prob[\textup{no abort}]
    \\
    &\geq \Bigl(1-\frac{m\sigma}{2\tau}\Bigr)\expect[\norm{e}_1 \mid \text{no abort}],
\end{aligned}
\end{equation}
where the last inequality uses
inequality (\ref{eq:validbound}) from \cref{sec:pf}.

Therefore, as $\tau \geq 2m\sigma$, we obtain $\expect[\norm{e}_1 \mid \text{no abort}] \leq m \sigma (1-\frac{1}{4})^{-1} \leq 2 m\sigma$, as required.
\end{proof}

\begin{theorem} \label{thm:rsp completeness}
If Alice and Bob follow the process given in \cref{fig:rsp protocol}, then the expected trace distance between $\ket{\alpha, b}$ and $\ket{\beta}$, conditioned on Alice not aborting, satisfies
\begin{equation}
    \expect[\norm{\ketbra{\alpha, b}-\ketbra{\beta}}_1 \mid \textup{no abort}] \leq \frac{4\pi m \sigma}{q},
\end{equation}
where the expectation is over the distribution on $(A,t,s,e,f,y)$ defined by \cref{fig:rsp protocol}.
\end{theorem}
\begin{proof}
Throughout this proof, we use the notation defined in \cref{fig:rsp protocol}. In step 3, by using arguments similar to those in our proof of \cref{thm:comp}, we see that
\begin{equation}
    \ket{\psi} = \frac{1}{\sqrt{2}}\bigl(\ket{x_1} \ket{1} + \ket{x_0}\ket{0}\bigr),
\end{equation}
where $x_0,x_1\in \mathbb{Z}_q^n$ and $x_1 = x_0 + s$, when  Alice does not abort.

Then, using \cref{prop:last_qubit}, we see that $\ket{\psi'} = \frac{1}{\sqrt{2}}\bigl(\ket{0} + e^{i\theta}\ket{1}\bigr)$, where 
\begin{equation}
\begin{aligned}
\theta \coloneqq & \frac{2\pi a \cdot (x_0-x_1)}{q}   + \pi ([x_0] - [x_1]) u = - \frac{2\pi a \cdot s}{q}   + \pi b  \mod 2\pi.
\end{aligned}
\end{equation}

Then, since $r \coloneqq 2\pi w/q = 2\pi(a \cdot s + f^\top e + \alpha)/q$, we deduce
\begin{align}
\ket{\beta} \coloneqq \begin{bmatrix}
1 & 0\\ 0 & e^{ir}
\end{bmatrix} \ket{\psi'} =     Z^b \frac{1}{\sqrt{2}}(\ket{0} + e^{i\phi}\ket{1}),
\end{align}
where $\phi \coloneqq\frac{2\pi f^\top e}{q} + \frac{2\pi \alpha}{q}$.

Therefore, the trace distance between $\ket{\alpha, b}$ and $\ket{\beta}$ is
\begin{equation}
    \norm{\ketbra{\alpha, b} - \ketbra{\beta}}_1 = 2(1-\abs{\braket{\alpha,b}{\beta}}^2)^{1/2} =  2 \abs{\sin(\frac{1}{2}\cdot \frac{2\pi f^\top e }{q})} \leq \frac{2\pi \abs{f^\top e}}{q},
\end{equation}
where the last inequality uses \cref{lem:sin square}. Therefore,
\begin{align*}
&\expect[\norm{\ketbra{\alpha, b} - \ketbra{\beta}}_1 \mid \textup{no abort}]
\\
\leq& \expect\Biggl[\frac{2\pi \abs{f^\top e}}{q} \Biggm| \textup{no abort}\Biggr]
\leq \frac{2\pi}{q} \expect \bigl[ \norm{e}_1 \bigm| \textup{no abort}\bigr] \leq \frac{4\pi m \sigma}{q},
\end{align*}
where the last inequality uses \cref{lem:no_abort_rsp}, as desired.
\end{proof}

\subsection{Blindness}
\label{subsec:blindness}

The purpose of this subsection is to prove \cref{thm:rsp_blindness}, which shows that a non-uniform quantum polynomial-time adversary can compute the value of $\alpha$ in Protocol~\textbf{P} in \cref{fig:rsp protocol} with at most negligible advantage over random guessing. In other words, such an adversary is blind to the value of $\alpha$. Intuitively, this is because the value $(a,w) = (f^\top A, f^\top v + \alpha)$ sent by Alice in her first message can be seen as an encryption of $\alpha \in \mathbb{Z}_q$ under the public key $(A,v)$.

To state \cref{thm:rsp_blindness}, we define the following distributions.
For $x\in \mathbb{Z}_q$, we define $\mathcal{D}_x$ to be the distribution on $\mathbb{Z}_q^{m\times n}\times \mathbb{Z}_q^m \times \mathbb{Z}_q^n\times \mathbb{Z}_q$ such that an element $(A,v,a,w)$ is sampled as follows: 
\begin{enumerate}
    \item $(A,t) \leftarrow \GenTrap()$,
    \item  $v \coloneqq As+e$, where $s \leftarrow \mathbb{Z}_q^n$ and $e \leftarrow G(\sigma,\tau)^m$,
    \item $a \coloneqq f^\top A$, where $f \leftarrow \{0,1\}^m$, and
    \item $w \coloneqq f^\top v + x$.
\end{enumerate}
We define $\tilde{\mathcal{D}}_x$ to be the same as $\mathcal{D}_x$ except with the first step replaced by $A\leftarrow \mathbb{Z}_q^{m\times n}$. We define $\mathcal{D}$ to be the distribution on $\mathbb{Z}_q^{m\times n}\times \mathbb{Z}_q^m \times \mathbb{Z}_q^n\times \mathbb{Z}_q$ such that an element $(A,v,a,w)$ is sampled as follows: 
\begin{equation}
    A\leftarrow \mathbb{Z}_q^{m\times n}, \quad v \leftarrow \mathbb{Z}_q^n, \quad a \leftarrow \mathbb{Z}_q^n, \quad \text{and} \quad w \leftarrow \mathbb{Z}_q.
\end{equation}

We can now state and prove \cref{thm:rsp_blindness}. The proof is essentially the same as that of \cref{prop:regev} but we include it for completeness.

\begin{theorem}[Blindness with respect to $\alpha$]\label{thm:rsp_blindness}
Let $\Guess \colon \mathbb{Z}_q^{m\times n} \times \mathbb{Z}_q^m \times \mathbb{Z}_q^n \times \mathbb{Z}_q \to D(\mathbb{Z}_q)$ be a non-uniform quantum polynomial-time algorithm. Suppose that the $\LWE_{n,q,G(\sigma,\tau)}$ problem is hard. Then for all $x,y\in \mathbb{Z}_q$, we have
\begin{equation}\label{eq:blindness}
\begin{aligned}
    &|\prob[\Guess(A,v,a,w) = x \mid (A,v,a,w) \leftarrow \mathcal{D}_y]
    \\
    &\quad - \prob[\Guess(A,v,a,w) = x \mid (A,v,a,w) \leftarrow \mathcal{D}] | \leq \negl.
\end{aligned}
\end{equation}
\end{theorem}

\begin{proof}
For two real functions of $\lambda$, $a = a(\lambda)$ and $b = b(\lambda)$, we write $a\simeq b$ to mean $\abs{a-b}\leq \negl$. Then, we have
\begin{align*}
    &\quad \prob[\Guess(A,v,a,w) = x \mid (A,v,a,w) \leftarrow \mathcal{D}_y]
    \\
    &\simeq \prob[\Guess(A,v,a,w) = x  \mid (A,v,a,w) \leftarrow \tilde{\mathcal{D}}_y]
    \\
    &\simeq \prob\Bigl[\Guess(A,v,a,w) = x \Bigm|  \substack{A\leftarrow \mathbb{Z}_q^{m\times n}, \, v \leftarrow \mathbb{Z}_q^n, \, f\leftarrow \{0,1\}^m, \\ a = f^\top A, \, w = f^\top v +y}\Bigr]
    \\
    &\simeq \prob[\Guess(A,v,a,w) = x \mid  A\leftarrow \mathbb{Z}_q^{m\times n}, v \leftarrow \mathbb{Z}_q^n, a \leftarrow \mathbb{Z}_q^n, u\leftarrow \mathbb{Z}_q, w = u + y]
    \\
    &= \prob[\Guess(A,v,a,w) = x \mid (A,v,a,w) \leftarrow \mathcal{D}],
\end{align*}
where the first approximation follows from \cref{prop:trap} (under our parameter settings in \cref{fig:parameters}), the second approximation follows from the LWE hardness assumption and $\Guess$ being a non-uniform quantum polynomial-time algorithm, and the third approximation follows from the leftover hash lemma (see Lemma 2.1 in \cite{alwen2011generating}). The theorem follows by the triangle inequality.
\end{proof}

\section{Future Directions}
\label{sec:future directions}

The technique of applying  rotated measurements to a claw-state seems to yield a higher degree of control in cryptographic protocols, and it invites applications to tasks beyond proofs of quantumness and remote state preparation.  In this section, we sketch initial steps towards three additional applications.  

\paragraph{\textbf{\textup{Verifiable Randomness.}}} Proofs of quantumness are naturally connected to
protocols for verifiable randomness. The paper \cite{brakerski2018cryptographic}
proved that certified randomness can be generated by a classical verifier and a quantum prover using only classical communication.  The subsequent paper \cite{mahadev2022} improved the rate of randomness generation.  
In~\cref{app:ver}, we take a step in the same direction by proving the following one-shot randomness expansion result.

\begin{theorem}
\label{thm:oneshotrand}
Suppose that the $\LWE_{n,q,G ( \sigma, \tau )}$ problem is hard.
Suppose that Bob's behavior in Protocol $\mathbf{Q}$ (\cref{fig:pfprot}) is such that $H_{min} ( d' \mid
pk, ct, b'=0 )  \leq  \delta$,
where $\delta$ is a positive real constant.  Then,
$\Pr[\textnormal{success}]
 \leq  \frac{3}{4} 
+ O ( \sqrt{\delta} ) 
+ \negl$.
\end{theorem}
The above theorem implies that if Bob achieves a score that is non-negligibly above $\frac{3}{4}$, then his output $d'$ on input $b' = 0$ must contain new randomness (conditioned on all of the information received from Alice).  
The proof of \cref{thm:oneshotrand} is an adapted form of the rewinding argument used to prove \cref{thm:pfsound}.  This result suggests trying to incorporate more techniques from~\cite{brakerski2018cryptographic} to prove a multi-shot randomness expansion result.  A possible next step would be to prove a version of \cref{thm:oneshotrand} expressed in terms of Renyi entropy rather than min-entropy.  We expect that a randomness expansion result could be proved in the same parameter range described in subsection \ref{subsec:parchoices}, thus offering a potential advantage over \cite{brakerski2018cryptographic,mahadev2022} (which use an exponentially large modulus).

\paragraph{\textbf{\textup{Semi-Quantum Money.}}} Quantum money was first introduced in the seminal paper by~\cite{wiesner1983} and has seen many iterations since. More recently, Radian and Sattath have proposed a semi-quantum money scheme \cite{radian2022semi} where the bank is completely classical. They show that any ``strong 1-of-2 puzzle'' can be made into a private semi-quantum money scheme and give a construction of a strong 1-of-2 puzzle. Informally, a 1-of-2 puzzle is a four-tuple of algorithms $(G,O,S,V)$ that can be used in the following interactive protocol between a classical referee R and a quantum player P. 
\begin{enumerate}
    \item R samples $(p,k) \leftarrow G(1^\lambda)$ and sends ``puzzle'' $p$ to P but not its ``key'' $k$.
    \item P samples $(o,\rho) \leftarrow O(p) $, where $o$ is a classical ``obligation'' string and $\rho$ is a quantum state, and sends $o$ to R.
    \item R samples $b'\leftarrow \{0,1\}$ and sends $b'$ to P.
    \item P computes $d'\coloneqq S(p,o,\rho,b')$ and sends $d'$ to R.
    \item R computes $V(p,k,o,b',d') \in \{0,1\}$ and says P succeeds if it equals $1$. 
\end{enumerate}
The completeness of a 1-of-2 puzzle is P's optimal success probability. The soundness error of a 1-of-2 puzzle is the optimal probability of succeeding given both $b'=0$ \emph{and} $b'=1$ under the same obligation. A strong 1-of-2 puzzle is a 1-of-2 puzzle with completeness 1 and soundness error 0.

A formal definition of a 1-of-2 puzzle is given in~\cref{app:money}. Under that definition, we construct a 1-of-2 puzzle $\mathbf{Z}$ in \cref{fig:puzzle} that is based on Protocol~$\mathbf{Q}$ (\cref{fig:pfprot}). We give the soundness and completeness error of $\mathbf{Z}$ in \cref{prop:quantum_money} and defer its proof to~\cref{app:money}. 

\begin{figure}[ht]
    \centering
    \fbox{\parbox{4.71in}{
    \normalsize
    \textbf{1-of-2 puzzle} $\mathbf{Z} = (G, O, S, V)$

    \begin{itemize}
    \item $G$. On receiving security parameter $1^\lambda$ as input, set parameters according to \cref{fig:parameters}. Sample $(pk = (A, v), s, t) \leftarrow \Gen_J()$ and $b \leftarrow \{ 0, 1 \}$. Compute  $ct = (a,w) \leftarrow \Encrypt_J( pk , b)$. Output $(p = (pk, ct), k = (s,t,b))$.
    \item $O$. On receiving $p = (pk = (A,v),ct = (a,w))$ as input, follow step 2 of~\cref{fig:prover} and output $(o = ( y, u_1, \cdots, u_{nQ}), \rho = L )$.
    \item $S$. On receiving $(p,o,L,b')$ as input, follow step 4 of \cref{fig:prover} and output $d'$.
    \item $V$. On receiving $p = (pk = (A,v),ct = (a,w))$, $k=(s,t,b)$, $o = ( y, u_1, \cdots, u_{nQ})$, $b'$, $d'$ as input, follow step 5 of~\cref{fig:pfprot} to obtain $d = d(p,k,o)$. Output 1 if $d \oplus d' = b \wedge b'$, otherwise output 0.
\end{itemize}
}}
    \caption{A 1-of-2 puzzle $\mathbf{Z}$ constructed from Protocol~$\mathbf{Q}$.}
    \label{fig:puzzle}
\end{figure}

\begin{proposition}\label{prop:quantum_money}
    Suppose that the $\LWE_{n,q,G ( \sigma, \tau )}$ problem is hard. Then $\mathbf{Z} = (G,O,S,V)$ defined in \cref{fig:puzzle} is a 1-of-2 puzzle with completeness $\cos^2 \left( \frac{\pi}{8} \right) - \frac{5m \sigma^2}{q^2}  - \frac{m \sigma}{ 2 \tau }$ and soundness error $\frac{1}{2} + \negl$.
\end{proposition}

As shown in \cite[Proof of Theorem 2]{radian2022semi}, to convert our 1-of-2 puzzle~$\mathbf{Z}$ in \cref{fig:puzzle} into a private semi-quantum money scheme, it suffices to convert it into a strong 1-of-2 puzzle. \cite[Corollary 21]{radian2022semi} gives a way of converting a 1-of-2 puzzle with completeness \emph{1} and soundness error equal to a constant less than 1 into a strong 1-of-2 puzzle via parallel repetition. However, their method does not immediately apply in our setting because our completeness is not 1. We conjecture instead that \emph{threshold} parallel repetition can be used to convert $\mathbf{Z}$ (or any 1-of-2 puzzle with a constant gap between its soundness error and completeness) into  a strong 1-of-2 puzzle. We give a precise formulation of our conjecture in~\cref{app:money}.

The main benefit of using our 1-of-2 puzzle construction is the size of the resulting quantum money state. In~\cite{radian2022semi}, the size of the money state in their semi-quantum money scheme is equal to the size of the quantum state $\rho$ produced by the obligation algorithm $O$ of the strong 1-of-2 puzzle. Their construction of a strong 1-of-2 puzzle is based on the parallel repetition of a 1-of-2 puzzle where the $\rho$ produced by $O$ is a \emph{claw state}. On the other hand, in our construction, the strong 1-of-2 puzzle is a (threshold) parallel repetition of a 1-of-2 puzzle where each $\rho$ produced by $O$ is a \emph{single qubit}. It is crucial to minimize the size of quantum money states since they need to be stored and protected against decoherence for long periods of time before transactions.

\paragraph{\textbf{\textup{Hamiltonian simulation with encryption.}}} While the results in this paper have focused on the preparation of single qubits, they invite extensions to other types of quantum operations.  We briefly sketch an approach that would enable a limited class of unitary operations on quantum registers of arbitrary size.\footnote{We thank Honghao Fu for helping us with the details in this subsection.}

Suppose that $w$ is a register, and let $W = \mathbb{C}^w$.  Suppose that $P$ is a Hermitian projection operator on $W$ such that there exists a polynomially sized unitary quantum circuit $U \colon W \to \mathbb{C}^2 \otimes W$ which maps $v \mapsto \left| 0 \right> \otimes v $ for all $v$ in the null space of $P$, and maps $v \mapsto \left| 1 \right> \otimes v$ for all $v$ in the support of $P$.

Suppose that Bob possesses register $w$ in state $\ell \in W$, and Alice wishes to have Bob apply the operation $e^{itP}$ to $w$ for some real value $t$ ($0 \leq t < 2 \pi$), without revealing the value of $t$ to Bob.  Alice chooses $\alpha \in \mathbb{Z}_q$ so that $\alpha / q \approx t$ and executes the first step of Protocol~$\mathbf{P}$ (\cref{fig:rsp protocol}).  Bob applies the quantum circuit $U$ to obtain a joint state $\ell' = \left| 0 \right> \otimes \ell_0 + \left| 1 \right> \otimes \ell_1$ of a pair of registers $(\zeta,w)$, where $\zeta$ is a qubit register and $P \ell_0 = 0$, $P \ell_1 = \ell_1$.  Then, Bob constructs the following state, which is the same as state (\ref{rspprepstate}) in Protocol~$\mathbf{P}$ except for the addition of a fourth register ($w$):
    \begin{equation*}
    \ket{\phi} \coloneqq 
    \frac{1}{\sqrt{2}\sqrt{q^n(2\tau+1)^{m}}}
        \sum_{\substack{x \in \mathbb{Z}_q^n}} \hskip0.1in \sum_{c \in \{ 0, 1 \}} \hskip0.1in
        \sum_{\substack{g \in \mathbb{Z}_q^m \\ \| g \|_\infty 
        \leq \tau }} \ket{x}\ket{c}
        \ket{Ax - cv + g} \ket{\ell_c}.
    \end{equation*}
Alice and Bob then follow the remaining steps of Protocol~$\mathbf{P}$, and Bob approximately obtains either the state
    \begin{equation}
        \frac{1}{\sqrt{2}} \bigl( \ket{0} \ket{\ell_0} + e^{i \, t } \ket{1} \ket{\ell_1} \bigr) 
        \quad 
        \text{or}
        \quad
        \frac{1}{\sqrt{2}} \bigl( \ket{0} \ket{\ell_0} - e^{i \, t } \ket{1} \ket{\ell_1} \bigr).
    \end{equation}
Lastly, Bob applies the reverse quantum circuit $U^{-1}$ to approximately obtain either $e^{itP} \ell$ or $e^{i(t+\pi)P} \ell$.

At this point, Bob has successfully carried out the unitary operator $e^{itP}$ (modulo a possible reflection $(\mathbb{I} - P)/2$) without knowing the value of $t$. This approach is easily generalizable with $P$ replaced by any Hermitian operator $H$ that has only two eigenvalues. A natural next step would be to explore whether a similar procedure could be carried out for arbitrary Hermitian operators $H$.

\section*{Acknowledgements}

We thank Gorjan Alagic, Alexandru Cojocaru, and Yi-Kai Liu for useful discussions and comments on an earlier version of this work. Y.A. acknowledges support from the Crown Prince International Scholarship Program and the Arabian Gulf University. A.M acknowledges support from the U.S. Army Research Office under Grant Number W911NF-20-1-0015 and AFOSR MURI project ``Scalable Certification of
Quantum Computing Devices and Networks'', as well as partial support by faculty startup grant from Virginia Tech and by the Commonwealth Cyber Initiative. D.W. acknowledges support from the Army Research Office (grant W911NF-20-1-0015); the Department of Energy, Office of Science, Office of Advanced Scientific Computing Research, Accelerated Research in Quantum Computing program; and the National Science Foundation (grant DMR-1747426).  This paper is partly a contribution of the National Institute of Standards and Technology.


\bibliography{references}
\bibliographystyle{alphaurl}

\newpage

\appendix

\section{Mathematical Lemmas}
\begin{lemma}
\label{lem:var}
For any $\sigma \in \mathbb{N},  \tau > 0$,
\begin{equation}
\expect [ X^2 \mid
X \leftarrow G ( \sigma, \tau ) ]
\leq \sigma^2.
\end{equation}
\end{lemma}

\begin{proof}
For notational convenience, for $x\in \mathbb{Z}$, we write
\begin{equation}
    \prob_1[x] \coloneqq \Pr[x \leftarrow G(\sigma)] \quad \text{and} \quad \prob_2[x] \coloneqq \Pr[x \leftarrow G(\sigma,\tau)].
\end{equation}
We also write
\begin{equation}
    \quad Z \coloneqq \sum_{x\in \mathbb{Z} : \abs{x} \leq \tau} \prob_1[x] \quad \text{and} \quad L \coloneqq \expect [ X^2 \mid X \leftarrow G ( \sigma, \tau ) ].
\end{equation}
Note that $L$ is the quantity we need to upper bound and
\begin{equation}\label{eq:gaussian_expected_note}
    L = \sum_{x\in \mathbb{Z} : \abs{x} \leq \tau} x^2 \prob_2[x] = \sum_{x\in \mathbb{Z} : \abs{x} \leq \tau} x^2 \frac{\prob_1[x]}{Z} \leq \tau^2.
\end{equation}

We upper bound $L$ as follows.
\begin{align*}
    L &= L(1-Z) + LZ
    \\
    &\leq \tau^2(1-Z) + \sum_{x\in \mathbb{Z} : \abs{x} \leq \tau} x^2 \prob_1[x] &&\text{(by \cref{eq:gaussian_expected_note})}
    \\
    &= \tau^2 \sum_{x\in \mathbb{Z} : \abs{x} > \tau} \prob_1[x] + \sum_{x\in \mathbb{Z} : \abs{x} \leq \tau} x^2 \prob_1[x]
    \\
    &\leq \sum_{x\in \mathbb{Z} : \abs{x} > \tau} x^2 \prob_1[x] + \sum_{x\in \mathbb{Z} : \abs{x} \leq \tau} x^2 \prob_1[x]
    \\
    &= \sum_{x\in \mathbb{Z}}x^2 \prob_1[x] 
    \\
    &= \expect[X^2 \mid X\leftarrow G(\sigma)],
\end{align*}
but, by \cref{lem:gaussian_properties}, we have
\begin{equation}
\expect[X^2 \mid X\leftarrow G(\sigma)]  = \variance[X \mid X\leftarrow G(\sigma)] \leq \sigma^2.
\end{equation}
Therefore, $L\leq \sigma^2$, as  required.
\end{proof}

\begin{lemma}
\label{lem:abs}
For any $\sigma \in \mathbb{N},  \tau > 0$,
\begin{equation}
\expect [ \, \left| X \right| \mid
X \leftarrow G ( \sigma, \tau ) ]
\leq \sigma.
\end{equation}
\end{lemma}

\begin{proof}
This lemma follows immediately from \cref{lem:var}.  We have
\begin{equation}
    \expect [\,  \left| X \right| \mid
X \leftarrow G ( \sigma, \tau ) ] \leq \sqrt{\expect [  \abs{X}^2  \mid
X \leftarrow G ( \sigma, \tau ) ]} \leq \sigma,
\end{equation}
as desired.
\end{proof}

\begin{lemma}\label{lem:sin square}
The following inequality holds for any real value t: $\left| \sin  t \right| \leq 
\left| t \right|$.
\end{lemma}

\begin{proof}
We have $\sin^2 t = \int_0^t \int_0^s 2 \cos(2r) \diff r \diff s \leq \int_0^t \int_0^s 2 (1) \diff r \diff s = t^2$.
\end{proof}

\begin{lemma}
\label{lem:sin}
The following inequality holds for any real values $t, s$:
\begin{equation}
\sin^2 (t+s) \leq \sin^2 s + t \sin ( 2 s )
+ t^2.
\end{equation}
\end{lemma}

\begin{proof} Let $f ( x ) = \sin^2 ( x + s)$.  We have $f ( 0 ) = \sin^2 s$, $f' ( 0 ) = \sin 2s$, and $f''(x ) = 2 \cos (2x + 2s)$.  If we let
\begin{equation}
g ( x ) = \sin^2 s + x \sin 2s + x^2 ,
\end{equation}
then $g(0) = f(0),$ $g'(0) = f'(0)$, and $g''(x) \geq f'' (x )$ for all $x$, which yields $g(x) \geq f (x )$ for all $x$ as desired. 
\end{proof}

\begin{lemma}\label{lem:sum}
Suppose that $X_1, \ldots, X_\ell$ are independent real-valued random variables and that $\expect[X_1] =\expect[X_2] = \ldots = \expect[X_{\ell-1}] = 0$.  Then,
\begin{equation}
\expect[ (X_1 + \ldots + X_\ell )^2 ]  = 
\expect[ X_1^2] + \expect[ X_2^2 ] +
\ldots + \expect[ X_\ell^2 ].
\end{equation}
\end{lemma}

\begin{proof}
We have
\begin{equation}
\expect[ (X_1 + \ldots + X_\ell )^2 ]  = 
\sum_i \expect[X_i^2 ] + 
\sum_{i \neq j} \expect [X_i] \expect[X_j].
\end{equation}
All terms in the second summation are clearly zero.
\end{proof}

\section{Proof of \texorpdfstring{\cref{prop:regev}}{Theorem 3.2}}
\label{app:regev}

Let $\Gen'()$ denote an algorithm that merely outputs a uniformly random pair $pk = (A, v)$ and does not output a secret key. By the LWE assumption, 
\begin{equation}\label{eq:breakquantity}
    \prob[b'=b \mid pk \leftarrow \Gen_K(), b\leftarrow\{0,1\},ct\leftarrow \Encrypt_K(pk,b), b'\leftarrow \mathcal{B}(pk,ct)]
\end{equation}
is negligibly different from
\begin{equation*}
\prob [ b' = b \mid pk \leftarrow \Gen' (), b \leftarrow \{ 0, 1 \},
ct \leftarrow \Encrypt_K ( pk, b ) ,
b' \leftarrow \mathcal{B} ( pk, ct )]
\end{equation*}
Meanwhile, the leftover hash lemma (see Lemma 2.1 in \cite{alwen2011generating}) implies that the distribution of $(A, v, f^\top A, f^\top v)$, when $f \in \{ 0, 1 \}^m, A \in \mathbb{Z}_q^{m \times n }, v \in \mathbb{Z}_q^m$ are all sampled uniformly, is itself negligibly close to uniform.  Therefore the quantity in \cref{eq:breakquantity} is also negligibly close to 
\begin{equation*}
\prob [ b' = b \mid pk \leftarrow \Gen' (), 
ct \leftarrow \mathbb{Z}_q^{n+1},  b \leftarrow \{ 0, 1 \},
b' \leftarrow \mathcal{B} ( pk, ct )].
\end{equation*}
Since $ct$ and $pk$ are independent of $b$ in this expression, the quantity above is obviously equal to $\frac{1}{2}$. This completes the proof.

\section{Trapdoors for LWE matrices}
\label{app:trap}

\begin{proof}[Proof of \cref{prop:trap}]
Our proof is essentially the same as the proof in \cite{micciancio-peikert2012trapdoors}. The main difference is that we use the infinity norm, rather than the Euclidean norm, to bound error vectors.  We define $\GenTrap()$ as follows.\\

\noindent Algorithm $\GenTrap()$:
\begin{enumerate}
    \item Let
    \begin{equation}
    g = \left[ \begin{array}{c} 1 \\
    2 \\ 4 \\ \vdots \\ 2^{Q-1} \end{array}
    \right]
    \end{equation} and let $G$ be the $nQ \times n$ matrix
    \begin{equation}
    G = \left[ \begin{array}{ccccc}
    g \\
    & g \\
    && g \\
    &&& \ddots \\
    &&&& g
    \end{array} \right].
    \end{equation}
    
    \item Sample a matrix
    $M \in \mathbb{Z}_q^{(Q+1)n \times n}$ with entries chosen uniformly from $\mathbb{Z}_q$, and a matrix
    $N \in \mathbb{Z}_q^{Qn \times (Q+1)n}$ with entries chosen uniformly from $\{ 0, 1 \}$.  Let
    \begin{equation}
    A  =  \left[ \begin{array}{c}
    G + NM \\ \hline M \end{array} \right].
    \end{equation}
    \item Let $t = N$.  Return $(A, t)$.
\end{enumerate}
\end{proof}

By the leftover hash lemma (see Lemma 2.1 in \cite{alwen2011generating}) if $z \in \{ 0, 1 \}^{(Q+1)n}$ is chosen uniformly at random, then the distribution of \[ \left[ \begin{array}{c} M \\ \hline\\[-10pt] z^\top M \end{array} \right] \] is within statistical distance $2^{-n/2}$ from uniformly random.  Iterating this fact, we find that the matrix
\[ \left[ \begin{array}{c} M \\ \hline  N M \end{array} \right] \] is within statistical distance $nQ 2^{-n/2}$ from uniformly random, and the same applies to the matrix $A$.  This proves the first claim of \cref{prop:trap}.

We sketch a method for the $\Invert$ algorithm (see \cite{micciancio-peikert2012trapdoors} for more details).
The algorithm $\Invert$ receives as input the matrix $A$, a vector $v \coloneqq As + e$ where $e$ satisfies $ \left\| e \right\|_\infty \leq 2 \tau$, and the trapdoor matrix $N$.  Let $v_1$ denote the vector consisting of the first $Qn$ entries of $v$, and let $v_2$ denote the vector consisting of the remaining entries of $v$.  We have
\begin{equation}
v_1  =  (G + NM) s + e_1 \quad \text{and} \quad
v_2  =  Ms + e_2,
\end{equation}
where $e_1, e_2$ have infinity norm upper bounded by $2 \tau$.  Letting
\begin{equation}
v' \coloneqq v_1 - Nv_2,
\end{equation}
we find
\begin{equation}
v' = Gs + (e_1 - Ne_2).
\end{equation}
Let $e' \coloneqq e_1 - Ne_2$.  Then, $v' = Gs + e'$, and $\left\| e' \right\|_\infty \leq 2\tau + (Q+1)n(2 \tau) < q/(2Q)$.  
Let
    \begin{equation}
    S \coloneqq \left[ \begin{array}{cccccc}
    2 & -1 &&&& \\
    & 2 & -1 &&&  \\
    && \ddots && &  \\
    &&& 2 & -1   \\
        \left[ q \right]_1 &
    \left[ q \right]_2 &
    \left[ q \right]_3 &
    \cdots &
    \left[ q \right]_Q
   \end{array} \right]
    \end{equation}
and let $Y$ be the $Q^2 \times Q^2$ matrix which consists of $Q$ diagonal blocks, each equal to $S$.  Then, $YG = 0$, and therefore if we compute $w \coloneqq Yv'$, we have
\begin{equation}
w  =  Y(Gs + e') = Ye'.
\end{equation}
Since each entry of $e'$ has absolute value less than $q/(2Q)$ and each row of $Y$ has trace-norm less than or equal to $Q$, the equation $Y e' = w$ can be solved for $e'$ simply by inverting the matrix Y over the real numbers.  Then, we compute $Gs = v' - e'$ and recover $s$. If no solution exists, we assume that Invert returns the vector $0^n$.

\section{Verifiable Randomness}\label{app:ver}

We recall that for any classical state $\gamma$ of a pair of registers $(\chi, \omega)$, the {min-entropy} of $\chi$ conditioned on $\omega$ is defined by
\begin{eqnarray}
H_{min} ( \chi \mid \omega )_\gamma :=  
- \log K,
\end{eqnarray}
where $K$ denotes the optimal probability with which a (computationally unlimited) party who possesses the register $\omega$ can guess the value of $\chi$.

In this section, we will prove \cref{thm:oneshotrand}. We assume, without loss of generality, that Bob's behavior has the form shown in \cref{fig:cheatquantum}.  We first prove the following lemma, which can be thought of as a form of the gentle measurement lemma (see, e.g., Lemma 9 in \cite{winter1999coding}).

\begin{figure}[ht]
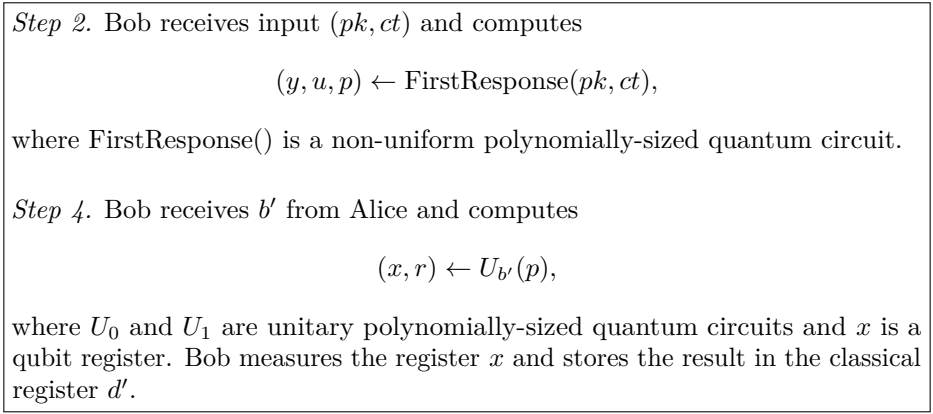

    \centering
    \fbox{\parbox{4.71in}{ 
    \normalsize
\textit{Step 2.}  Bob receives input $(pk, ct)$ and computes
\begin{equation*}
(y, u, p) \leftarrow \FirstResponse(pk, ct),
\end{equation*}
where $\FirstResponse()$ is a 
non-uniform polynomially-sized quantum circuit.  

\vskip0.2in

\textit{Step 4.} Bob receives $b'$ from Alice and computes
\begin{equation*}
(x, r) \leftarrow U_{b'} (p),
\end{equation*}
where $U_0$ and $U_1$ are unitary polynomially-sized quantum circuits and $x$ is a qubit register.  Bob measures the register $x$ and stores the result in the classical register $d'$.
}}
    \caption{A model of a quantum adversary
    for Protocol~\textbf{Q} in \cref{fig:pfprot}. The register $p$ may be in a quantum state after Step 2. We require that the circuits $U_0$ and $U_1$ are unitary circuits (i.e., involving no creation of registers or measurements).
}
    \label{fig:cheatquantum}
\end{figure}

\begin{lemma}
\label{lem:gentle}
Let $\zeta$ be a qubit register and let $c$ be a classical register.  Let $\rho$ be a classical-quantum state of $(c,\zeta)$, let $\rho'$ be the state that is obtained from $\rho$ by  measuring $\zeta$ in the computational basis, and suppose that
\begin{eqnarray}
H_{min} ( \zeta \mid c )_{\rho'} & \leq & \epsilon.
\end{eqnarray}
Then,
\begin{eqnarray}
\left\| \rho - \rho' \right\|_1 & \leq & O ( \sqrt{ \epsilon } ).
\end{eqnarray}
\end{lemma}

\begin{proof}
In the case in which $c$ is a trivial register, we can write $\rho$ in the form
\begin{eqnarray}
\rho & = & \left[ \begin{array}{cc}
\alpha & \beta \\ \overline{\beta} & \gamma 
\end{array} \right]
\end{eqnarray}
with $\left| \beta \right|^2 \leq \alpha \gamma$.  Since $\max \{ \alpha, \gamma \}$ is lower bounded by $1 - 2^{-\epsilon} = 1 - O ( \epsilon)$ by assumption, we have
\begin{eqnarray}
\left\| \rho - \rho' \right\|_1 & = & \left\| \left[ \begin{array}{cc}
0 & \beta \\ \overline{\beta} & 0
\end{array} \right] \right\|_1
 \\
& = & 2 \left| \beta \right| \\
& \leq & \sqrt{ (1 - O ( \epsilon ) ) O ( \epsilon ) } \\
& \leq & O ( \sqrt{ \epsilon }).
\end{eqnarray}
The general case follows from the concavity of the square root function.
\end{proof}

\begin{proof}[Proof of \cref{thm:oneshotrand}]
We consider the experiments shown in Figures~\ref{fig:t0}--\ref{fig:r1}.  The probability that Bob succeeds in Protocol~$\mathbf{Q}$ is the average of the following four quantities:
\begin{itemize}
    \item The probability that $d=d'$ in Experiment~$\mathbf{T}_0$.
    
    \item The probability that
    $d \neq d'$ in Experiment~$\mathbf{T}_1$.
    
    \item The probability that $d=d'_0$ in Experiment~$\mathbf{R}_0$.
    
    \item The probability that $d=d'_0$ in Experiment~$\mathbf{R}_1$.
\end{itemize}
We are therefore interested in proving bounds on these four quantities.

\begin{figure}[ht]
    \centering
    \fbox{\parbox{4.71in}{ \normalsize
\textbf{Experiment} $\mathbf{T}_0$: \\

\begin{enumerate}
\item Let $b := 0$.

\item Compute $(pk, s, t) \leftarrow \Gen_J ()$ and $ct \leftarrow \Encrypt_J (pk, b)$.  

\item Compute $(y, u, p ) \leftarrow \FirstResponse( pk, ct )$.

\item Compute $x_0 \coloneqq \Invert ( A, y, t)$
    and $x_1 \coloneqq \Invert ( A, y + v , t )$ and
    \begin{equation}
        d \coloneqq u \cdot ([x_0] \oplus [x_1]).
    \end{equation}

\item Compute $(x,r) \leftarrow U_1 ( p )$.

\item Let $b' := 1$.  Measure the qubit $x$ in the computational basis, and store the result in a bit register $d'$.
\end{enumerate}
}}
\caption{A version of the encrypted CHSH game with specific inputs: $b = 0, b'=1$.}
\label{fig:t0}
\end{figure}

\begin{figure}[ht]
    \centering
    \fbox{\parbox{4.71in}{ \normalsize
\textbf{Experiment} $\mathbf{T}_1$: \\

\begin{enumerate}
\item Let $b := 1$.  
\end{enumerate}

Steps $2-6$ are the same as in Experiment~$\mathbf{T}_1$.
}}
\caption{A version of the encrypted CHSH game with specific inputs: $b = 1, b'=1$.}
\label{fig:t1}
\end{figure}

\begin{figure}[ht]
    \centering
    \fbox{\parbox{4.71in}{ \normalsize
\textbf{Experiment} $\mathbf{R}_0$: \\

\begin{enumerate}
\item Let $b := 0$.

\item Compute $(pk, s, t) \leftarrow \Gen_J ()$ and $ct \leftarrow \Encrypt_J (pk, b)$.  

\item Compute $(y, u, p ) \leftarrow \FirstResponse( pk, ct )$.

\item Compute $x_0 \coloneqq \Invert ( A, y, t)$
    and $x_1 \coloneqq \Invert ( A, y + v , t )$ and
    \begin{equation}
        d \coloneqq u \cdot ([x_0] \oplus [x_1]).
    \end{equation}

\item Compute $(x, r) \leftarrow U_0 (p)$.

\item Measure the qubit $x$ in the computational basis, and store the result in a bit register $d'_0$.  

\item Compute $p \leftarrow U_0^{-1} ( x,r)$.

\item Compute $(x,r) \leftarrow U_1 ( p )$.

\item Measure the qubit $x$ in the computational basis, and store the result in a bit register $d'_1$.

\item Compute $b' \leftarrow \{ 0, 1 \}$ and $d' := d'_b$.

\end{enumerate}
}}
\caption{A rewinding experiment using the procedures from \cref{fig:cheatquantum}.}
\label{fig:r0}
\end{figure}

\begin{figure}[ht]
    \centering
    \fbox{\parbox{4.71in}{ \normalsize
\textbf{Experiment} $\mathbf{R}_1$: \\

\begin{enumerate}
\item Let $b:=1$.  
\end{enumerate}

Steps $2-10$ are the same as in Experiment~$\mathbf{R}_0$.
}}
\caption{A modified version of Experiment~$\mathbf{R}_0$.  The only difference is that $ct$ is an encryption of the bit $1$ instead of $0$.}
\label{fig:r1}
\end{figure}

We begin by considering Experiments~$\mathbf{R}_0$ and $\mathbf{R}_1$. Let $C_0$ denote the probability of the event $d'_0 \neq d'_1$ in Experiment~$\mathbf{R}_0$, and let $C_1$ denote the probability of the same event in Experiment~$\mathbf{R}_1$.
Note that the quantities $C_0$ and $C_1$ must be at most negligibly different, since otherwise one could use the procedure in Steps 3-10 of Experiment $\mathbf{R}_0$ to distinguish between an encryption of $b=0$ and an encryption of $b=1$ with non-negligible probability, thus violating Proposition~\ref{prop:regev2}.

For Experiment~$\mathbf{R}_0$, we have
\begin{equation}
\textnormal{Pr}_{\mathbf{R}_0} [ d=d'_0]
+ \textnormal{Pr}_{\mathbf{R}_0} [ d = d'_1 ] 
\leq  2 - \textnormal{Pr}_{\mathbf{R}_0} [ d'_0 \neq d'_1] =  2 - C_0,
\end{equation}

Similarly, for Experiment~$\mathbf{R}_1$, we have
\begin{equation}
\textnormal{Pr}_{\mathbf{R}_1} [ d=d'_0]
+  \textnormal{Pr}_{\mathbf{R}_1} [ d \neq d'_1 ] \leq  2 - \textnormal{Pr}_{\mathbf{R}_1} [ d'_0 = d'_1]
 =  1 + C_1.
\end{equation}
Note that in Experiment~$\mathbf{R}_0$, Lemma~\ref{lem:gentle} implies that the state of $(pk, ct, d, b, p)$ immediately before Step 5 is within trace distance $O ( \sqrt{\delta} )$ of the state of $(pk, ct, d, b, p)$ immediately before Step 8.  Comparing Experiment~$\mathbf{R}_0$ and Experiment~$\mathbf{T_0}$, we find that
\begin{eqnarray}
\left| \textnormal{Pr}_{\mathbf{R}_0}[d = d'_1]
- \textnormal{Pr}_{\mathbf{T}_0}[d = d'] \right|
& \leq & O ( \sqrt{\delta} ).
\end{eqnarray}
By similar reasoning,
\begin{eqnarray}
\left| \textnormal{Pr}_{\mathbf{R}_1}[d \neq d'_1]
- \textnormal{Pr}_{\mathbf{T}_1}[d \neq d'] \right|
& \leq & O ( \sqrt{\delta} ).
\end{eqnarray}
Therefore, the probability of success in Protocol~$\mathbf{Q}$ is upper bounded as follows:
\begin{eqnarray*}
&& \frac{1}{4} \left( \textnormal{Pr}_{\mathbf{T}_0}[d = d'] + 
\textnormal{Pr}_{\mathbf{R}_0}[d = d'_0] + 
\textnormal{Pr}_{\mathbf{R}_1}[d = d'_0] + 
\textnormal{Pr}_{\mathbf{T}_1}[d \neq d'] \right)
\\
& \leq & 
\frac{1}{4} \left( \textnormal{Pr}_{\mathbf{R}_0}[d = d'_1] + 
\textnormal{Pr}_{\mathbf{R}_0}[d = d'_0] + 
\textnormal{Pr}_{\mathbf{R}_1}[d = d'_0] + 
\textnormal{Pr}_{\mathbf{R}_1}[d \neq d'_1] + O ( \sqrt{\delta} ) \right) \\
& \leq & \frac{3}{4} + (C_1 - C_0)/2 + O ( \sqrt {\delta} ) \\
& \leq & \frac{3}{4} + \negl + O ( \sqrt { \delta } ),
\end{eqnarray*}
as desired.
\end{proof}

\section{Semi-quantum money}\label{app:money}

The formal definition of a 1-of-2 puzzle is as follows.
\begin{definition}[{1-of-2 puzzle~\cite[Definition 12]{radian2022semi}}]
A 1-of-2 puzzle is a four-tuple of efficient algorithms $ (G,O,S,V)$: the puzzle generator $G$, an obligation algorithm $O$, a 1-of-2 solver $S$, and a verification algorithm $V$. $G$ is a classical randomized algorithm, $V$ is a classical deterministic algorithm, and $O$ and $S$ are quantum algorithms.
\begin{enumerate}
    \item $G$ receives security parameter $1^\lambda$ as input and outputs a random puzzle $p$ and verification key $v$: $(p,k)\leftarrow G(1^\lambda)$.
    \item $O$ receives a puzzle $p$ as input and outputs a classical string $o$ called the obligation string and a quantum state $\rho$: $(o,\rho)\leftarrow O(p)$.
    \item $S$ receives $p,o,\rho$ and a bit $b\in \{0,1\}$ as input and outputs a classical string $a$: $a\leftarrow S(p,o,\rho,b)$.
    \item $V$ receives $p,k,o,b,a$ as input and outputs $0$ or $1$: $V(p,k,o,b,a)\in \{0,1\}$.
\end{enumerate}

Completeness. Let $\eta\colon \mathbb{N}\to[0,1]$. We say that a 1-of-2 puzzle has completeness $\eta$ if there exists a negligible function $\negl$ such that
\begin{equation}
    \prob[V(p,k,o,b,a) = 1] \geq \eta(\lambda) - \negl,
\end{equation}
where the probability is over $(p,k) \leftarrow G(1^\lambda)$, $(o,\rho) \leftarrow O(p)$, $b\leftarrow \{0,1\}$, and $a \leftarrow S(p,o,\rho,b)$.

Soundness (hardness). Let $h\colon \mathbb{N} \to [0,1]$. We say that the 1-of-2 puzzle $Z$ has soundness error $h$ if for any $\poly$-time quantum algorithm $T$, there exists a negligible function $\negl$ such that
\begin{equation}
    \prob[V(p,k,o,0,a_0) = V(p,k,o,1,a_1)  = 1] \leq h(\lambda) + \negl,
\end{equation}
where the probability is over $(p,k)\leftarrow G(1^\lambda)$ and $(o,a_0,a_1)\leftarrow T(p)$.
\end{definition}

We now prove \cref{prop:quantum_money}, which is restated below for convenience.
\begin{proposition}\label{prop:quantum_money_restated}
    Suppose that the $\LWE_{n,q,G ( \sigma, \tau )}$ problem is hard. Then $\mathbf{Z} = (G,O,S,V)$ defined in \cref{fig:puzzle} is a 1-of-2 puzzle with completeness $\cos^2 \left( \frac{\pi}{8} \right) - \frac{5m \sigma^2}{q^2}  - \frac{m \sigma}{ 2 \tau }$ and soundness error $\frac{1}{2} + \negl$.
\end{proposition}
\begin{proof}
    To prove correctness, observe that the probability of success of the player in $\mathbf{Z}$ is the same as the probability of success of the honest prover in the proof of quantumness protocol~$\mathbf{Q}$. This is because $G$ is simply Alice's actions in the first round of~$\mathbf{Q}$, $O$ and $S$ are Bob's actions in the first and second rounds of interaction in Protocol~$\mathbf{Q}$ respectively, and $V$ is the scoring function of Protocol~$\mathbf{Q}$. Additionally the challenge bit $b'$ is sampled uniformly randomly in both settings. Hence, \cref{thm:comp} implies that the completeness of $\mathbf{Z}$ is $\cos^2 \left( \frac{\pi}{8} \right) - \frac{5m \sigma^2}{q^2}  -
    \frac{m \sigma}{ 2 \tau }$.

    We now prove soundness. Let $T$ be a $\poly$-time quantum algorithm which receives a puzzle $p$ as input and outputs $(o, a_0, a_1)$, where $o$ is the obligation string and $a_0$ and $a_1$ are single bits. Showing the soundness claimed in the proposition is equivalent to showing
    \begin{equation}\label{eq:soundness_puzzle}
        \Pr[V(p,k,o,0,a_0) = V(p,k,o,1,a_1) = 1 ] \leq \frac{1}{2} + \negl,
    \end{equation}
    where the probability is over $(p,k) \leftarrow G(1^\lambda)$ and $(o,a_0,a_1) \leftarrow T(p)$.
    
    Assume for the sake of contradiction that
\begin{equation}\label{eq:money_assumption_contradiction}
        \Pr[V(p,k,o,0,a_0) = V(p,k,o,1,a_1) = 1] > \frac{1}{2} + \mu(\lambda),
    \end{equation}
    where the probability is over $(p,k) \leftarrow G(1^\lambda)$ and $(o,a_0,a_1) \leftarrow T(p)$, and $\mu$ is some non-negligible function.

    We can use $T$ to construct a $\poly$-time algorithm $B$ that breaks Regev's encryption scheme as follows.
    \begin{enumerate}
        \item $B$. Input: $(pk,ct)$. 
        \item $(o, a_0, a_1) \leftarrow T(p = (pk,ct))$
        \item Output: $a_0\oplus a_1$.
    \end{enumerate}

    Now, we have
    \begin{equation*}
    \begin{aligned}
        &\prob[B(pk,ct) = b \mid pk \leftarrow \Gen_J (), b \leftarrow \{ 0, 1 \}, ct \leftarrow \Encrypt_J ( pk, b )]
        \\
        =& \prob\Bigl[a_0 \oplus a_1 = b \bigm| \substack{pk \leftarrow \Gen_J (), b \leftarrow \{ 0, 1 \}, ct \leftarrow \Encrypt_J ( pk, b ),
        \\
        (o, a_0, a_1) \leftarrow T(pk,ct)}\Bigr] 
        \\
        =& \prob\Bigl[a_0 \oplus a_1 = b \bigm| (p=(pk,ct),k=(s,t,b)) \leftarrow G(1^\lambda),
        (o, a_0, a_1) \leftarrow T(pk,ct)\Bigr]
        \\
        \geq& \prob\Bigl[a_0 = d(p,k,o) \text{ and } a_1 = d(p,k,o)\oplus b \bigm| (p,k) \leftarrow G(1^\lambda),
        (o, a_0, a_1) \leftarrow T(p)\Bigr]
        \\
        =&  \Pr[V(p,k,o,0,a_0) = V(p,k,o,1,a_1) = 1 \mid (p,k) \leftarrow G(1^\lambda),
        (o, a_0, a_1) \leftarrow T(p)] 
        \\
        >& \frac{1}{2} + \mu(\lambda),
    \end{aligned}
    \end{equation*}
    where the first line follows from the definition of $B$, the second line follows from the definition of $G$, we recall the definition of $d(p,k,o)$ from \cref{fig:puzzle} in the third line, the fourth line follows from the definition of $V$, and the last line is our assumption in \cref{eq:money_assumption_contradiction}. 
    
    But $B$ is a $\poly$-time algorithm because $T$ is a $\poly$-time algorithm. Therefore, the last inequality contradicts \cref{prop:regev2} under the LWE hardness assumption made in the statement of the proposition. Therefore, \cref{eq:soundness_puzzle} must hold and the soundness of $\mathbf{Z}$ follows.
    
\end{proof}

It is straightforward to convert $\mathbf{Z}$ into a strong 1-of-2 puzzle --- and hence a semi-quantum money scheme via \cite{radian2022semi} --- assuming the following technical conjecture which shows that the soundness error decays exponentially with the number of repeats $\ell\in \mathbb{N}$ in threshold parallel repetition. This conjecture is known to hold with the word ``quantum'' (emphasized) replaced by ``randomized'', see \cite[Theorem 1.1]{impagliazzo2009chernoff}.\footnote{In \cite[Theorem 1.1]{impagliazzo2009chernoff}, the result is stated for a weakly verifiable puzzle instead of a 1-of-2 puzzle. However, the two types of puzzles have equivalent hardness, see \cite[Fact 19]{radian2022semi}.}

\begin{conjecture}\label{conjecture:quantum_money}
There exist constants $c_1,c_2>0$ such that the following holds. 

Let $\ell\in \mathbb{N}$, $\gamma>0$, and $\mathbf{Z} = (G,O,S,V)$ be a 1-of-2 puzzle. Suppose that $\mathbf{Z}$ is $\frac{1}{2}$-hard, i.e., for any $\poly$-time \emph{quantum} algorithm $T$, there exists a negligible function $\negl$ such that
\begin{equation}
    \prob[V(p,k,o,0,a_0) = V(p,k,o,1,a_1)  = 1] \leq \frac{1}{2} + \negl,
\end{equation}
where the probability is over $(p,k)\leftarrow G(1^\lambda)$ and $(o,a_0,a_1)\leftarrow T(p)$.

Then, for any $\poly$-time \emph{quantum} algorithm $\bar{T}$, we have
\begin{equation}
\begin{aligned}
    &\prob\bigl[\bigl|\{ i \in [\ell] \mid V(p_i,k_i,o_i,0,a_{0,i})= V(p_i,k_i,o_i,1,a_{1,i}) = 1\}\bigr|\bigr] \geq (1 + \gamma)\ell/2]
    \\
    \leq & \frac{c_1}{\gamma} \exp(-c_2 \,  \gamma^2 \, \ell),
\end{aligned}
\end{equation}
where the probability is over $(p_i, k_i)\leftarrow G(1^\lambda)$ for all $i\in [\ell]$ and $$(o_1,\dots,o_\ell,a_{0,1},\dots,a_{0,\ell},a_{1,1},\dots,a_{1,\ell}) \leftarrow \bar{T}(p_1,\dots,p_\ell, k_1,\dots,k_\ell).$$
\end{conjecture}

We now show that the following 1-of-2 puzzle $\mathbf{Z}^\ell$ defined from $\mathbf{Z}$ is a strong 1-of-2 puzzle assuming \cref{conjecture:quantum_money}.

\begin{definition}[Threshold parallel repetition of 1-of-2 puzzle $\mathbf{Z}$] Let $\mathbf{Z} = (G,O,S,V)$ be the 1-of-2 puzzle defined in \cref{fig:puzzle}. Let $\ell = \lambda \in \mathbb{N}$ and $\alpha \in (3/4, \cos^2(\pi/8))$. We define the 1-of-2 puzzle $\mathbf{Z}^\ell = (G^\ell, O^\ell,S^\ell,V^\ell)$ as follows.
\begin{enumerate}
\item $G^\ell$ receives security parameter $1^\lambda$ as input and runs $G(1^\lambda)$ $\ell$ times to sample $\ell$ random puzzles and their verification keys: $$((p_1,\dots,p_\ell),(k_1,\dots,k_\ell))\leftarrow G^n(1^\lambda),$$
where $(p_i,v_i)\leftarrow G(1^\lambda)$ for all $i\in [\ell]$.
\item $O^\ell$ receives $\ell$ puzzles as input and outputs $\ell$ obligation strings and $\ell$ quantum states: 
$$((o_1,\dots,o_\ell),(\rho_1\otimes\dots \otimes \rho_\ell))\leftarrow O^\ell(p_1,\dots,p_\ell),$$
where $(o_i,\rho_i)\leftarrow O(p_i)$ for all $i\in [\ell]$.
\item $S^\ell$ receives $\ell$ puzzles, $\ell$ obligation strings, $\ell$ quantum states, and a single bit $b\in \{0,1\}$ as input and outputs $\ell$ classical strings: 
$$(a_1,\dots,a_\ell)\leftarrow S^n((p_1,\dots,p_\ell),(o_1,\dots,o_\ell),\rho_1\otimes \dots \otimes \rho_\ell,b),$$
where $a_i \leftarrow S(p_i,o_i,\rho_i,b)$ for all $i\in [\ell]$.
\item $V^\ell$ receives $\ell$ puzzles, $\ell$ keys, $\ell$ obligation strings, a single bit $b\in \{0,1\}$, and $\ell$ answer strings as input and outputs $0$ or $1$: $$V^\ell((p_1,\dots,p_\ell),(k_1,\dots,k_\ell),(o_1,\dots,o_\ell),b,(a_1,\dots,a_\ell))\in \{0,1\}$$
is equal to 1 if and only if
\begin{equation}\label{eq:verification_threshold}
\bigl|\{ i \in [\ell] \mid V(p_i,k_i,o_i,b,a_i)= 1\}\bigr| \geq \alpha \ell.
\end{equation}
\end{enumerate}
\end{definition}

The definition of $\mathbf{Z}^\ell$ is the same as \cite[Definition 14]{radian2022semi} except the verification algorithm $V^\ell$ checks whether the number of successes exceeds a threshold, see~\cref{eq:verification_threshold}. 

Using $\alpha < \cos^2(\pi/8)$, a simple Chernoff bound implies that $\mathbf{Z}^\ell$ has completeness 1. Using $\alpha > 3/4$, so that $\gamma \coloneqq 4\alpha - 3 > 0$, we can show that $\mathbf{Z}^\ell$ has soundness error $0$ assuming \cref{conjecture:quantum_money} as follows. Let $\bar{T}$ be a $\poly$-time quantum algorithm, then
\begin{equation*}
\begin{aligned}
    &\prob[V^\ell(p,k,o,0,a_{0}) = V^\ell(p,k,o,1,a_1) = 1]
    \\ 
    =&\prob\bigl[\forall b \in \{0,1\}, \bigl|\{ i \in [\ell] \mid V(p_i,k_i,o_i,b,a_i)= 1\}\bigr|  \geq \alpha \ell\bigr]
    \\
    \leq& \prob\bigl[\bigl|\{ i \in [\ell] \mid V(p_i,k_i,o_i,0,a_{0,i})= V(p_i,k_i,o_i,1,a_{1,i}) = 1\}\bigr| \geq (2\alpha-1)\ell\bigr]
    \\
    =& \prob\bigl[\bigl|\{ i \in [\ell] \mid V(p_i,k_i,o_i,0,a_{0,i})= V(p_i,k_i,o_i,1,a_{1,i}) = 1\}\bigr| \geq (1+\gamma)\ell/2 \bigr]
    \\
    \leq& \frac{c_1}{\gamma} \exp(-c_2\, \gamma^2 \ell) 
    \\
    \leq& \negl,
\end{aligned}
\end{equation*}
where all probabilities are over $(p = (p_1,\dots, p_\ell),k = (k_1,\dots, k_\ell))\leftarrow G^\ell(1^\lambda)$ and $(o=(o_1,\dots,o_\ell),a_0 = (a_{0,1},\dots, a_{0,\ell}),a_1 = (a_{1,1},\dots, a_{1,\ell})) \leftarrow \bar{T}(p,k)$, the second-to-last line follows from \cref{conjecture:quantum_money}, and the last line uses $\ell = \lambda$.
\end{document}